\documentclass[twoside]{IEEEtran}

\usepackage{vkmacros}
\graphicspath{{figures/}}

\newif\ifmapx
{\catcode`/=0 \catcode`\\=12/gdef/mkillslash\#1{#1}}
\edef\jobnametmp{\expandafter\string\csname empirical_apx\endcsname}
\edef\jobnameapx{\expandafter\mkillslash\jobnametmp}
\edef\jobnameexpand{\jobname}
\ifx\jobnameexpand\jobnameapx
\mapxtrue
\else
\mapxfalse
\fi

\long\def\apxonly#1{\ifmapx{\color{blue}#1}\fi}

\psfrag{Rd}{\small{$R(d)$}}
\psfrag{k}{\small{$k$}}
\psfrag{n}{\small{$k$}}
\psfrag{R}{\small{$R$, bits/sample}}
\psfrag{V(d)}{\small{$\tilde{\mathcal V}(d)$, bits$^2$/sample}}
\psfrag{ABES}{\small{Achievability \cite[(207)]{kostina2011fixed} } }
\psfrag{CBES}{\small{Converse \cite[(203)]{kostina2011fixed} } }
\psfrag{NBES}{\small{Approximation \eqref{2ordernoisy}}}
\psfrag{NBESlog}{\small{Approximation \eqref{2ordernoisy}} $+ \frac {\log k} {2k}$}

\psfrag{ABESeq}{\small{Achievability \eqref{ABESeq}}}
\psfrag{CBESeq}{\small{Converse \eqref{CBESeq}}}
\psfrag{NBESeq}{\small{Approximation \eqref{2ordernoiseless}}}
\psfrag{NBESeqlog}{\small{Approximation \eqref{2ordernoiseless}} $+ \frac {\log k} {2k}$}

\title{Nonasymptotic noisy lossy source coding}

\author{%
Victoria Kostina$^{\ast}$, Sergio Verd\'{u}$^{\dag}$\\[0.5em]
{\small\begin{minipage}{\linewidth}\begin{center}
\begin{tabular}{ccc}
$^{\ast}$California Institute of Technology & \hspace*{0.5in} & $^{\dag}$Princeton University \\
Pasadena, CA 91125, USA && Princeton, NJ 08544, USA\\
\url{vkostina@caltech.edu} && \url{verdu@princeton.edu}
\end{tabular}
\end{center}\end{minipage}}
\thanks{
This work was supported in part by the National Science Foundation (NSF)
under Grant CCF-1016625 and by the Center for Science of Information
(CSoI), an NSF Science and Technology Center, under Grant CCF-0939370.

This work was presented in part at the 2013 IEEE Information
Theory Workshop \cite{kostina2013ITWnoisysc}. 
}
}

\IEEEoverridecommandlockouts
\begin{document}
\maketitle
\begin{abstract}
This paper shows new general nonasymptotic achievability and converse bounds and performs their dispersion analysis for the lossy compression problem in which the compressor observes the source through a noisy channel.  While this problem is asymptotically equivalent to a noiseless lossy source coding problem with a modified distortion function, nonasymptotically there is a noticeable gap in how fast their minimum achievable coding rates approach the common rate-distortion function, as evidenced both by the refined asymptotic analysis (dispersion) and the numerical results. The size of the gap between the dispersions of the noisy problem and the asymptotically equivalent noiseless problem depends on the stochastic variability of the channel through which the compressor observes the source. 

\end{abstract}


\section{Introduction}
Consider a lossy compression setup in which the encoder has access only to a noise-corrupted version $X$ of a source $S$, and we are interested in minimizing (in some stochastic sense) the distortion $\mathsf d(S, Z)$ between the true source $S$ and its rate-constrained representation $Z$ (see Fig. \ref{fig:noisysc}). This problem arises if the object to be compressed is the result of an uncoded transmission over a noisy channel, or if the observed data is subject to errors inherent to the measurement system. Examples include speech in a noisy environment and photographs corrupted by noise introduced by the image sensor. Since we are concerned with reproducing the original noiseless source rather than preserving the noise, the distortion measure is defined with respect to the clean source. 

The noisy source coding setting was introduced by Dobrushin and Tsybakov \cite{dobrushin1962addnoise}, who showed that, when the goal is to minimize the average distortion, the noisy source coding problem is asymptotically equivalent to a certain surrogate noiseless source coding problem.  Specifically, for a stationary memoryless source with single-letter distribution $P_{\mathsf S}$ observed through a stationary memoryless channel with single-input transition probability kernel $P_{\mathsf X | \mathsf S}$, the noisy rate-distortion function under a separable distortion measure is given by
\begin{align}
R(d) &= \min_{
\substack
{
P_{\mathsf Z| \mathsf X} \colon \\
\E{\mathsf d( \mathsf S, \mathsf Z)} \leq d\\
\mathsf S - \mathsf X - \mathsf Z
}
}
I(\mathsf X; \mathsf Z)\\
&=
\min_{
\substack
{
P_{\mathsf Z|\mathsf X} \colon \\
\E{\bar{\mathsf d}(\mathsf X, \mathsf Z)} \leq d
}
}
I(\mathsf X; \mathsf Z) \label{Rdnoisy},
\end{align}
where the surrogate per-letter distortion measure is
\begin{equation}
 \bar {\mathsf d}(a, b) = \E{ \mathsf d(\mathsf S, b) | \mathsf X = a}, \label{dbarintro}
\end{equation}
and  $\mathsf S - \mathsf X - \mathsf Z$ means that $\mathsf S$, $\mathsf X$ and $\mathsf Z$ form a Markov
chain in this order. 
The expression for $R(d)$ in \eqref{Rdnoisy} implies that, in the limit of infinite blocklengths, the problem is equivalent to a conventional (noiseless) lossy source coding problem where the distortion measure is the conditional average of the original distortion measure given the noisy observation of the source. Berger \cite[p.79]{berger1971rate} used the surrogate distortion measure \eqref{dbarintro} to streamline the proof of \eqref{Rdnoisy}. Witsenhausen \cite{witsenhausen1980indirect} explored the capability of distortion measures defined through conditional expectations such as in \eqref{dbarintro} to treat various so-called indirect rate distortion problems. Sakrison \cite{sakrison1968sourceencoding} showed that if both the source and its noise-corrupted version take values in a separable Hilbert space and the fidelity criterion is mean-squared error, then asymptotically, an optimal code can be constructed by first obtaining a minimum mean-square estimate of the source sequence based on its noisy observation, and then compressing the estimated sequence as if it were noise-free. Wolf and Ziv \cite{wolf1970transmissionnoisy} observed that Sakrison's result holds even nonasymptotically, namely, that the minimum average distortion achievable in one-shot noisy compression of the object $S$ can be written as 
\begin{equation}
D^\star(M) = \E{ |S - \E{S|X}|^2} + \inf_{\mathsf f, \mathsf c}\E{ | \mathsf c( \mathsf f(X) ) - \E{S|X}|^2 } \label{Dstarmse},
\end{equation}
where the infimum is over all encoders $\mathsf f \colon \mathcal X \mapsto \left\{1, \ldots, M \right\}$ and all decoders $\mathsf c \colon \left\{1, \ldots, M \right\} \mapsto \widehat {\mathcal M}$, and $\mathcal X$ and $\widehat {\mathcal M}$ are the alphabets of the channel output and the decoder output, respectively. It is important to note that the choice of the mean-squared error distortion  is crucial for the validity of the additive decomposition in \eqref{Dstarmse}. For vector quantization of a Gaussian signal corrupted by an additive independent Gaussian noise under weighted squared error distortion measure, Ayano\u{g}lu \cite{ayanoglu1990quantizationnoisy}  found explicit expressions for the optimum quantization rule. Wolf and Ziv's result in \eqref{Dstarmse} was extended to waveform vector quantization under weighted quadratic distortion measures and to autoregressive vector quantization under the Itakura-Saito distortion measure by Ephraim and Gray \cite{ephraim1988unified}, and by Fisher el al. \cite{fischer1990estimation}, who studied a model in which both encoder and decoder have access to the history of their past input and output blocks, allowing exploitation of inter-block dependence. Thus, the cascade of the optimal estimator followed by the optimal compressor achieves the minimum average distortion in those settings as well.  

For lossy compression of a discrete memoryless source (DMS) corrupted by discrete memoryless noise,   Weissman and Merhav \cite{weissman2002tradeoffs,weissman2002limited} studied the best attainable tradeoff between the exponential rates of decay of the probabilities that the codeword length and the cumulative distortion exceeds respective thresholds  \cite{weissman2002tradeoffs} and proposed a sequential coding scheme  for sequential compression of a source sequence corrupted by noise \cite{weissman2002limited}. The findings in \cite{weissman2002tradeoffs} imply, in particular, that the noisy excess-distortion exponent cannot be obtained from that of a clean surrogate source. Weissman \cite{weissman2004universally} went on to study universally attainable error exponents in noisy source coding.

Under the logarithmic loss distortion measure \cite{courtade2014multiterminal}, the noisy source coding problem reduces to the information bottleneck problem \cite{tishby1999bottleneck}. Indeed, in the information bottleneck method, the goal is to minimize $I(X; Z)$ subject to the constraint that $I(S; Z)$ exceeds a certain threshold. Both problems are equivalent because the noisy rate-distortion function under logarithmic loss  is given by \eqref{Rdnoisy} in which $\E{\bar{\mathsf d}(\mathsf X, \mathsf Z)}$  is replaced by $H(\mathsf S| \mathsf Z)$. The solution to the information bottleneck problem thus becomes the answer to a fundamental limit, namely, the asymptotically minimum achievable noisy source coding rate under logarithmic loss.

In this paper, we give new nonasymptotic achievability and converse bounds for the noisy source coding problem, which generalize the noiseless source coding bounds in \cite{kostina2011fixed}. We evaluate those bounds numerically in special cases to illuminate their tightness 
for all but very small blocklengths, a regime for which Shannon's random selection of reproduction
points ceases to be near-optimal. We give a refined asymptotic analysis of those bounds.  Specifically, as in \cite{kostina2011fixed,ingber2011dispersion}, we fix the tolerable probability $\epsilon$ of exceeding distortion $d$ and we study how fast the rate-distortion function can be approached as the length of the data block we are coding over increases.\footnote{We refer the reader to \cite{kostina2011fixed} for a discussion of differences between the fidelity criteria of excess and average distortion in the nonasymptotic regime.} This regime is different from the large deviations asymptotics studied in \cite{weissman2002tradeoffs} in which a rate strictly greater than the rate-distortion function is fixed, and the excess distortion probability vanishes exponentially with $k$.  
For noiseless source coding of stationary memoryless sources with separable distortion measure, we showed previously \cite{kostina2011fixed} (see also \cite{ingber2011dispersion} for an alternative proof in the finite alphabet case) that the minimum number $M$ of representation points compatible with a given probability $\epsilon$ of exceeding distortion threshold $d$ at blocklength $k$ satisfies
\begin{equation}
\log M^\star(k, d, \epsilon) = k R(d) + \sqrt {k \mathcal V(d)}\Qinv{\epsilon} + \bigo{\log k} \label{2ordernoiseless},
\end{equation}
where $\mathcal V(d)$ is the rate-dispersion function, explicitly identified in \cite{kostina2011fixed}, and $\Qinv{\cdot}$ denotes the inverse of the complementary standard Gaussian cdf.  In this paper, we show that for noisy source coding of a discrete stationary memoryless source over a discrete stationary memoryless channel under a separable distortion measure, $\mathcal V(d)$ in \eqref{2ordernoiseless} is replaced by the noisy rate-dispersion function $\tilde {\mathcal V}(d) $, which can be expressed as
\begin{equation}
 \tilde {\mathcal V}(d)  = {\mathcal V}(d) +  \lambda^{\star 2} \Var{ \sd(\mathsf S, \mathsf Z^\star) | \sX, {\mathsf Z}^\star} \label{Vnoisydecomp} ,
\end{equation}
where 
   $\lambda^\star = - R^\prime(d)$, $\mathcal V(d)$ is the rate-dispersion function of the surrogate rate-distortion setup,  $\mathsf Z^\star$ denotes the reproduction random variable that achieves the rate-distortion function \eqref{Rdnoisy}, and the {\it conditional variance} of random variable $U$ conditioned on $V$ is defined as
\begin{equation}
\Var{ U | V} \triangleq \E{(U - \E{U | V})^2}. 
\end{equation}

 Note that $\tilde{\mathcal V}(d)$ cannot be expressed solely as a function of the source distribution and the surrogate distortion measure.  Thus, the noisy coding problem is, in general, not equivalent to the noiseless coding problem with the surrogate distortion measure.  Essentially, the reason is that taking the expectation in \eqref{dbarintro} dismisses the randomness introduced by the noisy channel, which cannot be neglected in the  analysis of the dispersion. Indeed, as seen from \eqref{Vnoisydecomp}, the difference between $\mathcal V(d)$ and $\tilde {\mathcal V}(d)$ is due to the stochastic variability of the channel from $S$ to $X$. 

The rest of the paper is organized as follows. After introducing the basic definitions in Section \ref{sec:defn}, we proceed to show new general nonasymptotic converse and achievability bounds in Sections \ref{sec:C} and \ref{sec:A}, respectively, along with their asymptotic analysis in Section \ref{sec:2order}. Finally, the example of a binary source observed through an erasure channel is discussed in Section \ref{sec:BES}. 

\begin{figure}
\includegraphics[width=1\linewidth]{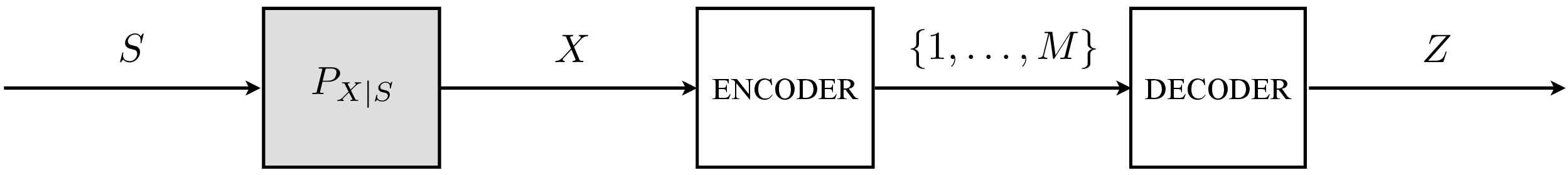}
\caption{Noisy source coding.}
\label{fig:noisysc}
\end{figure}

\section{Definitions}
\label{sec:defn}
Consider the one-shot setup in Fig. \ref{fig:noisysc} where we are given the distribution $P_S$ on the alphabet $\mathcal M$ and the transition probability kernel $P_{X|S} \colon \mathcal M \to \mathcal X$ describing the statistics of the noise corrupting the original signal. We are also given the distortion measure $\mathsf d \colon \mathcal M \times \widehat {\mathcal M} \mapsto [0, + \infty]$, where $\widehat {\mathcal M}$ is the representation alphabet. An $(M, d, \epsilon)$ code is a pair of random mappings 
$P_{U|X} \colon \mathcal X \mapsto \left\{1, \ldots, M \right\}$ and $P_{Z|U} \colon \left\{1, \ldots, M \right\} \mapsto \widehat {\mathcal M}$ such that $\Prob{ \mathsf d(S, Z) > d} \leq \epsilon$. 

Define
\begin{equation}
\mathbb R_{S, X}(d) \triangleq \inf_{
\substack
{
P_{ Z| X} \colon \\
\E{\bar{\mathsf d}( X,  Z)} \leq d
}
}
I( X;  Z), \label{RR(d)noisy}
\end{equation}
where
$\bar {\mathsf d} \colon \mathcal X \times \widehat {\mathcal M} \mapsto [0, +\infty]$ is given by
\begin{equation}
 \bar {\mathsf d}(x, z) \triangleq \E{ \mathsf d( S, z) |  X = x} \label{dbar}.
\end{equation}

We assume that $\mathbb R_{ S,  X}(d) < \infty$ for some $d$, and we denote
\begin{equation}
 d_{\min} \triangleq \inf \left\{ d\colon ~ \mathbb R_{ S,  X}(d) < \infty \right\} \label{dmin}.
\end{equation}

As in \cite{kostina2011fixed}, assume that the infimum in \eqref{RR(d)noisy} is achieved by some $P_{Z^\star | X}$ such that the constraint is satisfied with equality. Noting that this assumption guarantees differentiability of $\mathbb R_{S, X}(d)$, denote
\begin{align}
\lambda^\star &= - \mathbb R^\prime_{S, X}(d).
\end{align}

\begin{defn}[noisy $\mathsf d$-tilted information] \label{defn:dtilted}
For $d > d_{\min}$, the noisy $\mathsf d$-tilted information in $s \in \mathcal M$ given observation $x \in \mathcal X$ and representation $z \in \widehat {\mathcal M}$ is defined as
 \begin{align}
 \tilde{\jmath}_{S, X}(s, x, z, d) &\triangleq  
 \imath_{X; Z^\star}(x; z) + \lambda^\star \sd(s, z) - \lambda^\star d, \label{dtiltednoisy}
\end{align}
where $P_{Z^\star | X}$ achieves the infimum in \eqref{RR(d)noisy}, and the information density is denoted by
\begin{equation}
\imath_{X; Z^\star}(x; z) \triangleq \log \frac{dP_{ Z^\star| X = x}}{d P_{Z^\star}} (z).
\end{equation}
\end{defn}

As we will see, the intuitive meaning of the noisy $\mathsf d$-tilted information is the number of bits required to represent $s$ within distortion $d$ given observation $x$ and representation $z$. 

Taking the expectation of \eqref{dtiltednoisy} with respect to $S$ conditioned on $X$, we obtain $\jmath_X(x, d)$, the $\mathsf d$-tilted information in $x$ for the surrogate noiseless source coding problem:
\begin{align}
\jmath_X(x, d) 
&= \imath_{X; Z^\star}(x; z) + \lambda^\star \bar{\mathsf d}(x, z) - \lambda^\star d \label{dtiltedinfo},
\end{align}
where  
\eqref{dtiltedinfo} holds for $P_{Z^\star}$-a.e. $z$ (\!\cite[Lemma 1.4]{csiszar1974extremum}, \cite[Theorem 2.1]{kostina2013thesis}), which is a consequence of the fact that $Z^\star$ achieves the minimum in the convex optimization problem \eqref{RR(d)noisy}. 

Comparing \eqref{dtiltednoisy} and \eqref{dtiltedinfo}, we observe that
\begin{align}
\tilde \jmath_{S, X}(s, x, z, d) &=  \jmath_X(x, d) + \lambda^\star \sd(s, z) - \lambda^\star \E{\mathsf d(S, z) | X = x}  \label{dtiltedinfoa}
\end{align}

 From \eqref{dtiltednoisy} and \eqref{dtiltedinfoa}, we obtain
\begin{align}
\mathbb R_{S, X}(d) &= \E{ \tilde{\jmath}_{S, X}(S, X, Z^\star, d)}\\
&=  \E{ {\jmath}_{X}(X, d)}. 
\end{align}

%

To conclude the introduction to noisy $\sd$-tilted information, we discuss its relationship to a variation in rate-distortion function due to a small variation in probability distribution $P_{S X Z^\star}$. Assume that the alphabets $\mathcal M$, $\mathcal X$ and $\widehat {\mathcal M}$ are finite. 
Fix $d$. In order to rigorously define a derivative of $\mathbb R_{S, X}(d)$ with respect to $P_{S X Z^\star}$, we consider a finite measure on $\mathcal M \times \mathcal X \times \widehat {\mathcal M}$:
\begin{equation}
 Q_{ S  X  Z} = Q_{ S}(s) Q_{ X |  S}(x | s) Q_{ Z |  X}(z | x),
\end{equation}
which is not necessarily a probability measure. 
Furthermore, using the probability measures: 
\begin{align}
 P_{\bar S}(s) &= \frac{Q_{ S}(s)}{\sum_s Q_{ S}(s)}\\
 P_{\bar X | \bar S}(x | s) &= \frac{Q_{ X |  S}(x|s)}{\sum_{x \in \mathcal X} Q_{ X |  S}(x|s)}\\
 P_{\bar Z | \bar X}(z | x) &= \frac{Q_{ Z |  X}(z|x)}{\sum_{z \in \widehat{\mathcal M}} Q_{ Z |  X}(z|x)},
\end{align}
we introduce the following function of $Q_{ S  X  Z}$: 
\begin{equation}
 F(Q_{ S  X  Z}, d) \triangleq I(\bar X; \bar Z) + \lambda(P_{\bar S\bar X}) \left( \E{\sd (\bar S, \bar Z)} -  d \right), 
\end{equation}
where $\lambda(P_{\bar S\bar X}) = - \mathbb R^\prime_{\bar S, \bar X}(d)$ is a differentiable function of $P_{\bar S \bar X} = P_{\bar S} P_{\bar X | \bar S}$. In particular, $ F(P_{\bar S\bar X\bar Z^\star})  = \mathbb R_{\bar S, \bar X}(d)$, where $\bar Z^\star$ attains $R_{\bar S, \bar X}(d)$.
Denote the partial derivatives with respect to the coordinate at $(s, x, z)$: 
\begin{equation}
\dot{\mathbb R}_{S, X}(s, x, z, d) \triangleq   \left. \frac{\partial F(Q_{ S  X  Z}, d) }{\partial  Q_{ S  X  Z}(s, x, z) }   \right |_{Q_{ S  X  Z} = P_{S X Z^\star}} \label{eq:fdot}
\end{equation}

The following theorem demonstrates that the noisy $\mathsf d$-tilted information exhibits properties similar to those of the $\mathsf d$-tilted information listed in \cite[Theorem 2.2]{kostina2013thesis} (reproduced in Theorem \ref{thm:ingber} in Appendix \ref{appx:aux}). 
\begin{thm}
Assume that $\mathcal M$, $\mathcal X$ and $\widehat {\mathcal X}$ are finite sets, and let $d > d_{\min}$.  
Then,
\begin{align}
\dot{\mathbb R}_{S, X}(s, x, z, d)
 &= \tilde \jmath_{S, X}(s, x, z, d) - \mathbb R_{S, X}(d) \label{nsc:dif2},\\
 \Var{ \dot{\mathbb R}_{S, X}(S, X, Z^\star, d)} &= \Var{\tilde \jmath_{S, X}(S, X, Z^\star, d)}. \label{nsc:ingber}
\end{align}
where the variances are with respect to 
$(S,X, Z^\star) \sim P_{S} P_{X|S} P_{Z^\star | X}$. 
\label{thm:nsc:ingber}
\end{thm}
\begin{proof}
See Appendix \ref{appx:ingber}. 
\end{proof}

\apxonly{
WE DON"T REALLY NEED THIS. 

For a given distribution $P_{\bar Z}$ on $\widehat {\mathcal M}$ and $\lambda > 0$ define the transition probability kernel 
\begin{equation}
dP_{\bar Z^\star |X = x}(z) = \frac{dP_{\bar Z}( z) \exp\left( -  \lambda \bar{\mathsf d}( x,  z)\right) }{\E{\exp\left( - \lambda \bar {\mathsf d}( x,  \bar Z)\right)}} \label{PZ|Xstar}
\end{equation}

and define the function 
 \begin{align}
 \tilde{J}_{\bar Z}(s, x, \lambda) &\triangleq  
 D(P_{\bar Z^\star | X = x} \| P_{\bar Z}) + \lambda \bar {\mathsf d}_{\bar Z^\star}(s |x) \label{dtiltednoisyg}\\
 &= J_{\bar Z}(x, \lambda)  +  \lambda \bar {\mathsf d}_{\bar Z^\star}(s |x) - \lambda \E{\bar {\mathsf d}_{\bar Z^\star}(S |x) | X = x} 
\end{align}
where 
\begin{equation}
 J_{\bar Z}(x, \lambda) \triangleq \log \frac 1 {\E{ \exp\left( -\lambda \mathsf d(x, \bar Z)\right) } } \label{jtiltedg}
\end{equation}
where the expectation is with respect to the unconditional distribution of $P_{\bar Z}$. Similar to \cite[(2.26)]{kostina2013thesis}, we refer to the function
\begin{equation}
  \tilde{J}_{\bar Z}(s, x, \lambda) - \lambda d
\end{equation}
as the {\it generalized noisy $\mathsf d$-tilted information}. As in the noiseless case, the generalized $\mathsf d$-tilted information turns out to be relevant to the following optimization problem. 
\begin{equation}
\mathbb R_{S, X; \bar Z}(d) \triangleq \inf_{
\substack
{
P_{ Z| X} \colon \\
\E{\bar{\mathsf d}( X,  Z)} \leq d
}
}
D( P_{Z|X} \| P_{\bar Z} | P_X) \label{RR(d)g}
\end{equation}
}

\section{Converse bounds}
\label{sec:C}

Our main converse result for the nonasymptotic noisy lossy compression setting is the lower bound on the excess distortion probability as a function of the code size given in Theorem \ref{thm:Cnoisy}.  
For fixed $P_X$ and an auxiliary conditional distribution $P_{\bar X | \bar Z}$, denote the function 
\begin{align}
 f_{\bar X | \bar Z}(s, x, z) &\triangleq  \imath_{\bar X|\bar Z \| X}(x; z)  + \sup_{\lambda \geq 0} \lambda( \mathsf d(s, z) -  d) 
 - \log M \label{eq:faux},
\end{align}
where
\begin{align}
 \imath_{ \bar X| \bar Z\| X}( x;  z) &\triangleq \log \frac{dP_{\bar X| \bar Z = z}}{dP_{X}}(x) \label{ibar}.
\end{align}
Our converse result can now be stated as follows. 
\begin{thm}[Converse]
 If an $(M,d,\epsilon)$ code exists, then the following inequality must hold:
\begin{align}
\epsilon &\geq \inf_{ \substack { P_{Z|X} \colon\\ \mathcal X \mapsto \widehat{\mathcal M}}}  \sup_{\substack{  P_{\bar X | \bar Z} \colon \\ \widehat{\mathcal M} \mapsto \mathcal X }} \sup_
{\gamma \geq 0} \left\{ \Prob{ f_{\bar X | \bar Z}(S, X, Z) \geq \gamma} - \exp(-\gamma) \right\} \label{Cnoisy},
\end{align}
where the middle supremum is over those $P_{\bar X | \bar Z}$ such that Radon-Nikodym derivative of $P_{\bar X | \bar Z = z}$ with respect to $P_X$ at $x$ exists for $P_{Z|X} P_X$-a.e. $(z, x)$. 
\label{thm:Cnoisy}
\end{thm}
\begin{proof}
Let the encoder and decoder be the random transformations $P_{U|X}$ and $P_{Z|U}$, respectively, where $U$ takes values in $\{1, \ldots, M\}$. 
We have, for any $\gamma \geq 0$,
\begin{align}
&~
\Prob{f_{\bar X | \bar Z}(S, X, Z) \geq \gamma}
\notag\\
=&~ \Prob{  f_{\bar X | \bar Z}(S, X, Z) \geq \gamma, d(S, Z) > d} 
\notag\\
+
&~
 \Prob{ f_{\bar X | \bar Z}(S, X, Z) \geq \gamma, d(S, Z) \leq d} \\
 =&~ 
 \Prob{  d(S, Z) > d} \notag\\
 +&~  \Prob{ \imath_{\bar X|\bar Z \| X}(X; Z)  \geq \log M + \gamma, d(S, Z) \leq d} \label{-Csup}
 \\
 \leq&~ \epsilon + \Prob{ \imath_{\bar X|\bar Z \| X}(X; Z)  \geq \log M + \gamma}\\
 \leq&~\epsilon + \frac{\exp(-\gamma)}{M}\E{ \exp\left( \imath_{\bar X|\bar Z \| X}(X; Z) \right) } \label{-Ca}\\
\leq&~ \epsilon  \label{-Cb}\\
 +&~ \frac{\exp(-\gamma)}{M}   \sum_{u = 1}^M  \int_{z \in \widehat {\mathcal M}} dP_{Z|U}(z|u)   \int_{x \in \mathcal X} dP_{\bar X|\bar Z}(x|z)    \notag\\
=&~ \epsilon + \exp\left(-\gamma\right), 
\end{align}
where
\begin{itemize}
\item \eqref{-Csup} is by direct solution for the supremum; 
\item \eqref{-Ca} is by Markov's inequality; 
 \item \eqref{-Cb} follows by writing out the expectation on the left side with respect to $X - U - Z$ and upper-bounding
\begin{equation}
 P_{U|X}(u|x) \leq 1
\end{equation}
 for every $(x, u) \in \mathcal M \times \left\{ 1, \ldots, M\right\}$. 
 \end{itemize}
Finally, \eqref{Cnoisy} follows by choosing  $\gamma$ and $P_{\bar X| \bar Z}$ that give the tightest bound and $P_{Z|X}$ that gives the weakest bound in order to obtain a code-independent converse. 
\end{proof}
The following bound reduces the intractable (in high dimensional spaces) infimization over all conditional probability distributions $P_{Z|X}$ in Theorem \ref{thm:Cnoisy} to the infimization over the symbols of the output alphabet.   

\begin{cor}
 Any $(M,d,\epsilon)$ code must satisfy
\begin{align}
\epsilon &\geq \sup_{\gamma\geq 0} \sup_{\substack{  P_{\bar X | \bar Z} \colon \\ \widehat{\mathcal M} \mapsto \mathcal X }} \bigg\{ \E{ \inf_{z \in \widehat {\mathcal M}} \Prob{ f_{\bar X | \bar Z}(S, X, z) \geq \gamma | X} } 
\notag\\ &
- \exp(-\gamma)\bigg\} \label{Cnoisya},
\end{align}
and the supremum is over those $P_{\bar X | \bar Z}$ such that Radon-Nikodym derivative of $P_{\bar X | \bar Z = z}$ with respect to $P_X$ at $x$ exists for every $z \in \widehat {\mathcal M}$ and $P_X$-a.e. $x$. 
\label{cor:Cnoisya}
\end{cor}

\begin{proof}
We weaken \eqref{Cnoisy} using
\begin{align}
&~
 \inf_{P_{Z|X}} \Prob{ f_{\bar X | \bar Z}(S, X, Z) \geq \gamma } 
\notag\\
= &~ \E{ \inf_{P_{Z|X}} \Prob{ f_{\bar X | \bar Z}(S, X, Z) \geq \gamma | X} }\\
= &~ \E{ \inf_{z \in \widehat {\mathcal M}} \Prob{ f_{\bar X | \bar Z}(S, X, z) \geq \gamma | X} },
\end{align}
where we used $S - X - Z$. 
\end{proof}

In our asymptotic analysis, we will apply Corollary \ref{cor:Cnoisya} with suboptimal choices of $\lambda$ and $P_{\bar X | \bar Z}$. Note that if $\widehat {\mathcal M} = {\hat {\mathcal S}}^k$, where $\hat {\mathcal S}$ is a finite set, and $P_{\bar X | \bar Z}$ is chosen so that $P_{\bar X | \bar Z = z^k}$ is the same for all $z^k$ of the same type, the computation of the weakened version of the bound in \eqref{Cnoisya} has only polynomial-in-$k$ complexity. 

\begin{remark}
If $\mathrm{supp}(P_{Z^\star}) = \widehat {\mathcal M}$ and $P_{X|S}$ is the identity mapping so that ${\mathsf d}(S, z) = \bar{\mathsf d}(X, z)$ almost surely, for every $z$, then Corollary \ref{cor:Cnoisya} reduces to the noiseless converse in \cite[Theorem 7]{kostina2011fixed} by using \eqref{dtiltedinfo} after weakening \eqref{Cnoisya} with $P_{\bar X | \bar Z} = P_{X | Z^\star}$ and $\lambda = \lambda^\star$. 
\end{remark}

\begin{remark}
In many important special cases, the value of the probability in \eqref{Cnoisya} is the same for all $z$, obviating the need to perform the optimization over $z$ in \eqref{Cnoisya}. Indeed, weakening \eqref{Cnoisya} with $P_{\bar X | \bar Z} = P_{X | Z^\star}$ and $\lambda = \lambda^\star$ and using \eqref{dtiltedinfo}, 
we re-write the expression inside the conditional on $X = x$ probability in \eqref{Cnoisya} as
\begin{align}
 &~\imath_{X; Z^\star}(x; z)  +  \lambda^\star( \mathsf d(S, z) -  d) \notag\\
 =&~ \jmath_X(x, d) + \lambda^\star \left( \mathsf d(S, z) - \bar{\mathsf d}(x, z) \right) \label{eq:Cnoisyasym},
\end{align}
which according to \eqref{dtiltedinfo} holds for all $z \in \widehat {\mathcal M}$ if $\mathrm{supp}(P_{Z^\star}) = \widehat {\mathcal M}$. 
It follows that the conditional (on $X = x$) distribution of \eqref{eq:Cnoisyasym} is uniquely determined by the conditional distribution of the zero-mean random variable $ {\mathsf d}(S, z) - \bar{\mathsf d}(x, z)
$. Therefore, as long as conditional (on $X = x$, for all $x \in \mathrm{supp}(P_{X})$) distribution of $ {\mathsf d}(S, z) - \bar{\mathsf d}(x, z)
$ does not depend on the choice of $z \in \widehat {\mathcal M}$, any choice of $z$ in \eqref{Cnoisya} is as good as any other.  The condition is met, for example, by lossy coding of a Gaussian memoryless source observed through an AWGN channel under the mean-squared error distortion, as well as by coding of an equiprobable source observed through a symmetric channel under a symbol error rate constraint. \label{rem:symmetry}
\end{remark}

\section{Achievability bounds}
\label{sec:A}
Continuing in the single-shot setup, in this section we show the existence of codes of a given size, whose probability of exceeding a given distortion level is bounded explicitly.
\begin{thm}[Achievability]
 For any $P_{\bar Z}$ defined on $\widehat {\mathcal M}$, there exists a deterministic $(M, d, \epsilon)$ code with \footnote{We use the notation $\mathbb P^M[\cdot] = \left(\mathbb P[\cdot]\right)^M$.}
\begin{equation}
\epsilon \leq \int_0^1 \E{ \mathbb P^M\left[ \pi(X, \bar Z) > t  | X  \right] } dt \label{Anoisy},
\end{equation}
where $P_{X \bar Z} = P_X P_{\bar Z}$, 
\begin{equation}
\pi(x, z) = \Prob{\mathsf d(S, z) > d | X = x} \label{pi}. 
\end{equation} 
\label{thm:Anoisy}
\end{thm}
\begin{proof}
The proof invokes a random coding argument. Given $M$ codewords $(c_1, \ldots, c_M)$, the encoder $\mathsf f$ and decoder $\mathsf c$ achieving minimum excess distortion probability attainable with the given codebook operate as follows. Having observed $x \in \mathcal X$, the optimum encoder chooses 
\begin{equation}
i^\star \in \arg\min_{i} \pi(x, c_i), 
\end{equation}
with ties broken arbitrarily, so $\mathsf f(x) = i^\star$ and the decoder simply outputs $\mathsf c(\mathsf f(x)) = c_{i^\star}$. 

The excess distortion probability achieved by the scheme is given by
\begin{align}
\Prob{\mathsf d(S, \mathsf c(\mathsf f(X)) ) > d} &= \E{\pi(X, \mathsf c(\mathsf f(X)) )}\\
&= \int_0^1  \Prob{\pi(X, \mathsf c(\mathsf f(X)) ) > t} dt \\
&= \int_0^1  \E{\Prob{ \pi(X, \mathsf c(\mathsf f(X)) ) > t | X }} dt \label{-Anoisya}.
\end{align}

Now, we notice that
\begin{align}
1\left\{ \pi(x, \mathsf c(\mathsf f(x)) )  > t \right\} &=
1\left\{\min_{i \in 1, \ldots, M} \pi(x, c_i) > t \right\}  \\
&= \prod_{i = 1}^M 1 \left\{ \mathsf \pi(x, c_i) > t\right\},
\end{align}
and we average \eqref{-Anoisya} with respect to the codewords $Z_1, \ldots, Z_M$ drawn i.i.d. from $P_{\bar Z}$, independently of any other random variable, so that 
$P_{X Z_1 \ldots Z_M}  = P_{X} \times P_{\bar Z} \times \ldots \times P_{\bar Z}$, 
to obtain
\begin{align}
 &~
 \int_0^1 \E{ \prod_{i = 1}^M \Prob{ \pi(X, Z_i) > t | X } } dt 
 \\
 = &~ \int_0^1 \E{ \mathbb P^M\left[ \pi(X, \bar Z) > t  | X  \right] } dt \label{-Anoisyb}.
\end{align}
Since there must exist a codebook achieving excess distortion probability below or equal to the average over codebooks, \eqref{Anoisy} follows. 
\end{proof}
\begin{remark}
Notice that we have actually shown that the right side of \eqref{Anoisy} gives the exact minimum excess distortion probability of random coding, averaged over codebooks drawn i.i.d. from $P_Z$. 
\end{remark}

\begin{remark}
In the noiseless case, $S = X$, so almost surely
\begin{equation}
 \pi(X, z) = 1\left\{\mathsf d(S, z) > d\right\},
\end{equation}
 and the bound in Theorem \ref{thm:Anoisy} reduces to the noiseless random coding bound in \cite[Theorem 10]{kostina2011fixed}. 
\end{remark}
The bound in \eqref{Anoisy}, which in itself can be difficult to compute or analyze, can be weakened to obtain the following result in Corollary \ref{cor:AnoisyShannon}. Corollary \ref{cor:AnoisyShannon} generalizes \cite[Theorem 1]{kostina2011fixed} which distills Shannon's method for noiseless lossy compression in the single-shot setup. 

\begin{cor}
For any $P_{Z|X}$, there exists an $(M, d, \epsilon)$ code with 
\begin{align}
 \epsilon &\leq \inf_{\gamma \geq 0} \bigg\{ \Prob{d(S, Z) > d} + \Prob{ \imath_{X; Z}(X; Z) > \log M - \gamma} \notag \\ &
 + e^{-\exp(\gamma)} \bigg\} \label{AnoisyShannon},
\end{align}
where $P_{SXZ} = P_{S} P_{X|S} P_{Z|X}$. 
\label{cor:AnoisyShannon}
\end{cor}
\begin{proof}
Fix  $\gamma \geq 0$ and transition probability kernel $P_{Z|X}$. Let $P_X \to P_{Z|X} \to P_Z$ (i.e. $P_Z$ is the marginal of $P_{X} P_{Z|X}$), and let $P_{X \bar Z} = P_X P_Z$. We use the nonasymptotic covering lemma \cite[Lemma 5]{verdu2012multiuser} to establish 
\begin{align}
&~
\E{ \mathbb P^M\left[ \pi(X, \bar Z) > t  | X  \right] }  
\notag\\
\leq&~\Prob{ \pi(X, Z) > t }  
+ \Prob{ \imath_{X; Z}(X; Z) > \log M - \gamma } + e^{-\exp(\gamma)} \label{-Anoisycover}.
\end{align}
Applying \eqref{-Anoisycover} to \eqref{-Anoisyb} and noticing that
\begin{align}
 \int_0^1  \Prob{ \pi(X, Z) > t } dt &= \E{\pi(X, Z) } \\
 &= \Prob{\mathsf d(S, Z) > d},
\end{align}
we obtain \eqref{AnoisyShannon}. 
\end{proof}
While an analysis of the bound in Corollary \ref{cor:AnoisyShannon} for memoryless sources leads to a dispersion term which is larger than that in  \eqref{Vnoisydecomp}, it is an easy-to-compute bound which also leads to a simple proof of the achievability of the asymptote in \eqref{Rdnoisy}.

\apxonly{
THIS SEEMS USELESS.
\begin{thm}[Achievability, generalized $\mathsf d$-tilted information]
Suppose that $P_{ Z| X}$ is such that almost surely
\begin{align}
\mathsf d(S, Z) 
 &=  \bar {\mathsf d}_{ Z}(S |X)  \label{AdtiltedAss1}
 \end{align}
Then there exists an $(M,d,\epsilon)$ code with
\begin{align}
 \epsilon \leq \inf_{\gamma, \beta, \delta, P_{\bar Z}}  \Big\{
 &~
 \mathbb E \Big [ \inf_{ \lambda > 0} \Big\{
  \Prob{I = 0 |X} 
  + \Prob{\bar {\mathsf d}_{Z}(S|X) > d |X} 
  \notag\\
  &~
  +
 \left| 1 - \beta \exp \left(- \lambda \delta\right) \Prob{d - \delta \leq \bar {\mathsf d}_{Z}(S | X) \leq d |X, I = 1} \right|^+ | 
   \Big\}
   \Big]  + e^{-\frac M \gamma} 
 \Big\} \label{Adtiltedg}
\end{align} 
where 
\begin{equation}
  I(s, x) \triangleq 1\left\{ D(P_{Z | X = x} \| P_{\bar Z}) + \lambda \bar {\mathsf d}_{Z}(s |x) - \lambda d \leq \log \gamma - \log \beta  \right\}
\end{equation}
\label{thm:Adtiltedg}
\end{thm}

\begin{proof}
The bound in \eqref{Anoisy} implies that for an arbitrary $P_{\bar Z}$, there exists an $(M, d, \epsilon)$ code with 
\begin{align}
\epsilon
\leq
 &~
\int_0^1 \E{ \mathbb P^M\left[ \pi(X, \bar Z) > t  | X  \right] } dt  
\notag \\
\leq
&~
  e^{-\frac M \gamma}  \E {\min \left\{ 1, \gamma \int_0^1 \Prob{ \pi \left(X, \bar Z \right) \leq t  | X } dt \right\} } 
 +
   \int_0^1 \E{ \left| 1 - \gamma \Prob{ \pi(X, \bar Z) \leq t  | X } \right|^+} dt \label{Adtilted1}\\
   \leq
&~
  e^{-\frac M \gamma} 
 +
   \int_0^1 \E{ \left| 1 - \gamma \Prob{ \pi(X, \bar Z) \leq t  | X } \right|^+} dt
\end{align}
where to obtain \eqref{Adtilted1} we applied \cite{polyanskiy2012notes}
\begin{equation}
 (1-p)^M \leq e^{-Mp} \leq e^{-\frac M \gamma} \min(1, \gamma p) + |1 - \gamma p|^+
\end{equation}


Denote for brevity
 \begin{align}
\pi(x) &\triangleq \Prob{\bar {\mathsf d}_{Z}(S|x) > d | X = x},\\
g(s, x) &\triangleq  D(P_{Z | X = x} \| P_{\bar Z}) + \lambda \bar {\mathsf d}_{Z}(S |x) - \lambda d \label{gsk}
\end{align}

The first term in the right side of \eqref{Adtilted1} is upper bounded using the following chain of inequalities. 
\begin{align}
 &~
 \int_0^1 \left| 1 - \gamma \Prob{ \pi(X, \bar Z) \leq t  | X = x} \right|^+ 
 \notag\\
 \leq 
 &~
 \int_0^1 \left| 1 - \gamma \E{\exp\left( - \imath_{Z| X\| \bar Z}(x; Z)  \right) 1\left\{\pi(x, Z) \leq t \right\} | X = x} \right|^+ \\
 =
 &~
  \int_0^1 \left| 1 - \gamma 1\left\{ \pi(x) \leq t \right\} \E{\exp\left( -  \imath_{Z| X\| \bar Z}(x; Z) \right) } \right|^+ \label{Adtilted2} \\  
   =
 &~
 \pi(x) + (1 - \pi(x)) \left| 1 - \gamma \E{\exp\left( -  \imath_{Z| X\| \bar Z}(x; Z) \right) } \right|^+ \\
\leq
 &~
 \pi(x) + (1 - \pi(x)) \left| 1 - \gamma \exp\left( -  D(P_{Z|X = x} \| P_{\bar Z} ) \right) \right|^+ \label{Adtilted3}\\
 \leq
 &~
 \pi(x) + \left| 1 - \gamma \exp\left( -  D(P_{Z|X = x} \| P_{\bar Z} ) \right) \Prob{ \bar {\mathsf d}_{Z}(S | x) \leq d }\right|^+ \\
  \leq
 &~
 \pi(x) + \left| 1 - \gamma \exp\left( -  D(P_{Z|X = x} \| P_{\bar Z} ) \right) \Prob{ d - \delta \leq \bar {\mathsf d}_{Z}(S | x) \leq d }\right|^+ \\
  \leq
 &~
    \pi(x) 
  + 
  \left| 1 - \gamma  \E{\exp\left( -  
  g(S, x)  - \lambda \delta \right) \1 {d - \delta \leq \bar {\mathsf d}_{Z} (S|x) \leq d} }\right|^+ 
   \label{Adtilted4} \\
        \leq
 &~
    \pi(x) 
  + 
 \E{ \left| 1 - \gamma  \E{\exp\left( -  
  g(S, x)  - \lambda \delta \right) \1 {d - \delta \leq \bar {\mathsf d}_{Z} (S|x) \leq d} | I } \right|^+  } 
   \label{Adtilted4a} \\
           \leq
 &~
    \pi(x) 
  + 
     \Prob{ I = 0 | X = x} +
 \E{ \left| 1 - \gamma  \E{\exp\left( -  
  g(S, x)  - \lambda \delta \right) \1 {d - \delta \leq \bar {\mathsf d}_{Z} (S|x) \leq d} | I = 1} \right|^+  } 
   \\
     \leq
  &~
  \pi(x)  
  + 
   \Prob{ I = 0 | X = x}
    + 
    \left| 1 - \beta \exp \left(- \lambda \delta\right) \Prob{d - \delta \leq \bar {\mathsf d}_{Z}(S | x) \leq d \, | \, I = 1 } \right|^+
  \label{Adtilted4b}\\
  \leq
  &~
  \pi(x)  
  + 
   \Prob{ I = 0 | X = x } 
   \notag\\
    + 
    &~
    \left| 1 - \beta \exp \left(- \lambda \delta\right) \Prob{d - \delta  \leq \bar {\mathsf d}_{Z}(S | x) \leq d  \, | I = 1, X = x } \right|^+
  \label{Adtilted5}
\end{align}

\begin{itemize}

 \item in \eqref{Adtilted2} 
we observed using \eqref{AdtiltedAss1} that almost surely 
\begin{equation}
\pi(X, Z) = \pi(X) ;
\end{equation}

\item \eqref{Adtilted3} and \eqref{Adtilted4a} are by Jensen's inequality;

\item to obtain \eqref{Adtilted5}, we bounded 
\begin{align}
\gamma \exp (- g(S, x)   ) \geq 
\begin{cases}
 \beta &\text{ if } I = 1\\
   0 &\text{ otherwise }\\
\end{cases}
\end{align}

\end{itemize}

Taking the expectation of \eqref{Adtilted5} and recalling \eqref{Adtilted1}, \eqref{Adtiltedg} follows. 
\end{proof}

}

The following weakening of Theorem \ref{thm:Anoisy} is tighter than that in Corollary \ref{cor:AnoisyShannon} and is amenable to a tight second-order analysis.

\begin{thm}[Achievability]
For any $P_{\bar Z}$, there exists an $(M, d, \epsilon)$ code with 
\begin{align}
\epsilon \leq \Prob{  g_{\bar Z}(X, U) \geq \log \gamma  }  +  e^{-\frac M \gamma}  \label{Anoisy1},
\end{align} 
where $U$ is uniform on $[0, 1]$, independent of $X$, and \footnote{We understand that $g_{\bar Z}(x, t) = + \infty$ if $\nexists P_Z \colon \pi(x, Z) \leq t \text{ a.s. } $}
\begin{equation}
g_{\bar Z}(x, t) \triangleq  
 \inf_{P_Z \colon \pi(x, Z) \leq t \text{ a.s.} }  D(P_{Z} \|  P_{\bar Z} ). 
\label{Anoisy1}
\end{equation}

\label{thm:Anoisy1}
\end{thm}

\begin{proof}
The bound in \eqref{Anoisy} implies that for an arbitrary $P_{\bar Z}$, there exists an $(M, d, \epsilon)$ code with 
\begin{align}
\epsilon
\leq
 &~
\int_0^1 \E{ \mathbb P^M\left[ \pi(X, \bar Z) > t  | X  \right] } dt  
\notag \\
\leq
&~
  e^{-\frac M \gamma}  \E {\min \left\{ 1, \gamma \int_0^1 \Prob{ \pi \left(X, \bar Z \right) \leq t  | X } dt \right\} } 
\notag \\
 +
 &~
   \int_0^1 \E{ \left| 1 - \gamma \Prob{ \pi(X, \bar Z) \leq t  | X } \right|^+} dt \\
   \leq
&~
  e^{-\frac M \gamma} 
\notag \\
 +
 &~
   \int_0^1 \E{ \left| 1 - \gamma \Prob{ \pi(X, \bar Z) \leq t  | X } \right|^+} dt \label{Adtilted1},
\end{align}
where $| a |^+ \triangleq \max\{0, a\}$, and  to obtain \eqref{Adtilted1} we applied \cite{polyanskiy2012notes}
\begin{equation}
 (1-p)^M \leq e^{-Mp} \leq e^{-\frac M \gamma} \min(1, \gamma p) + |1 - \gamma p|^+.
\end{equation}
To bound $\left| 1 - \gamma \Prob{ \pi(X, \bar Z) \leq t  | X = x} \right|^+$, we let $P_Z$ be some distribution (chosen individually for each $x$ and $t$) such that $\pi(x; Z)  \leq t$ a.s., and we write
\apxonly{E.g. the distribution that achieves $R_x(t)$ will do but any other is ok too. }
\begin{align}
 &~
\left| 1 - \gamma \Prob{ \pi(X, \bar Z) \leq t  | X = x} \right|^+ 
 \notag\\
 \leq 
 &~
 \left| 1 - \gamma \E{\exp\left( - \imath_{Z \| \bar Z}(Z)  \right) 1\left\{\pi(x, Z) \leq t \right\} | X = x} \right|^+ \label{Acm}\\
 =
 &~
\left| 1 - \gamma \E{\exp\left( -  \imath_{Z \| \bar Z}(Z) \right) } \right|^+ \label{Adtilted2} \\
   \leq
 &~
  \left| 1 - \gamma \exp\left( -  D(P_{Z} \|  P_{\bar Z} ) \right)  \right|^+ \label{Adtilted3}\\
 \leq 
 &~
 1\{ D(P_{Z} \|  P_{\bar Z} ) > \log \gamma \} \label{Adtilted4},
\end{align}
where \eqref{Acm} is by the change of measure, 
\begin{align}
 \imath_{ \bar Z \| Z}( z) &\triangleq \log \frac{dP_{\bar Z}}{ dP_Z } (z)  \label{ibarZZ},
\end{align}
\eqref{Adtilted2} is by the choice of $P_Z$, \eqref{Adtilted3} is by Jensen's inequality, and \eqref{Adtilted4} follows from
\begin{align}
\gamma \exp\left( -  D(P_{Z} \|  P_{\bar Z} )\right)   \geq 
\begin{cases}
 1 &\text{ if } D(P_{Z} \|  P_{\bar Z} )  \leq \log \gamma \\
   0 &\text{ otherwise. }\\
\end{cases}
\end{align}
\end{proof}

\apxonly{
How we will apply it: 
Choose $P_{Z_t | X = x}$ equiprobable on the conditional type that achieves $R_{x, \bar Z}(d + \delta)$ where $\delta = \sqrt{\frac{V}{k}} \Qinv{t}$. Argue that $\pi(x; Z_t)  \leq t$ and that
$D(P_{Z_t | X = x} \|  P_{\bar Z} )   \approx R_{x, \bar Z}(d + \delta) \approx R_{x, \bar Z}(d) + \lambda^\star \delta$. Integrate wrt  $t$. }

\section{Asymptotic analysis}
\label{sec:2order}
In this section, we pass from the single shot setup of Sections \ref{sec:C} and \ref{sec:A} to a block setting by letting the alphabets be Cartesian products $\mathcal M = \mathcal S^k$, $\mathcal X = \mathcal A^k$, $\widehat {\mathcal M} = \hat {\mathcal S}^k$, and we study the second order asymptotics in $k$ of $M^\star(k, d, \epsilon)$, the minimum achievable number of representation points compatible with the excess distortion constraint $\Prob{ \mathsf d(S^k, Z^k) > d } \leq \epsilon$.  We make the following assumptions. 


\begin{enumerate}[(i)]

\item $P_{S^k X^k} = P_{\mathsf S}P_{\mathsf X | \mathsf S} \times \ldots \times  P_{\mathsf S} P_{\mathsf X | \mathsf S} $ and 
\begin{equation}
\mathsf d(s^k, z^k) = \frac 1 k \sum_{i = 1}^k \mathsf d(s_i, z_i).
\end{equation} \label{item:1}
\item The alphabets $\mathcal S$, $\mathcal A$, $\hat {\mathcal S}$ are finite sets.  
\item The distortion level satisfies $d_{\min} < d < d_{\max}$, where $d_{\min}$ is the minimum $d \geq 0$ such that $\mathbb R_{\mathsf S, \mathsf X}(d)$ is finite (as in \eqref{dmin}),
 and $d_{\max} =\inf_{\mathsf z \in \hat {\mathcal S}} \E{{\mathsf d}(\mathsf S, \mathsf z)}$, where the expectation is with respect to the unconditional distribution of $\mathsf S$. 
\item Let  $P_{\mathsf S \bar {\mathsf X} } = P_{\mathsf S | \mathsf X} P_{\bar {\mathsf X}}$. In a neighborhood of $P_{\mathsf X}$, $\mathrm{supp}(P_{\bar {\mathsf Z}^\star}) = \mathrm{supp}(P_{{\mathsf Z}^\star})$, where $\mathbb R_{\mathsf  S, \bar {\mathsf  X}}(d) = I(\bar{\mathsf  X}; \bar{\mathsf  Z}^\star)$, and $\mathbb R_{\mathsf S, \bar {\mathsf X}}(d)$ is twice continuously differentiable as a function of $P_{\bar {\mathsf X}}$. \label{item:5}
\end{enumerate}

The following result is obtained via a second order analysis of Corollary \ref{cor:Cnoisya} (Appendix \ref{appx:2ordernoisyC}) and Theorem \ref{thm:Anoisy1} (Appendix \ref{appx:2ordernoisyA}).
\begin{thm}[Gaussian approximation]
For $0 < \epsilon < 1$, under assumptions \eqref{item:1}--\eqref{item:5}, the minimum
number of representation points compatible with a given probability $\epsilon$ of exceeding distortion $d$ at blocklength $k$
satisfies
\begin{align}
\log M^\star(k, d, \epsilon) &= k R(d) + \sqrt {k \mathcal {\tilde V}(d)}\Qinv{\epsilon} + \bigo{\log k} \label{2ordernoisy},
\end{align}
where
\begin{align}
\tilde{\mathcal V}(d) &= \Var{\tilde{\jmath}_{\mathsf S; \mathsf X}(\mathsf S, \mathsf X, \mathsf Z^\star, d)} \label{Vtilde}
\end{align}
is the noisy rate-dispersion function. 
\label{thm:2ordernoisy}
\end{thm}
The rate-dispersion function of the surrogate noiseless problem is given by (see \cite[(83)]{kostina2011fixed})
\begin{equation}
\mathcal V(d) = \Var{\jmath_{\mathsf X}(\mathsf X, d)}, 
\end{equation}
where $\jmath_{\mathsf X}(\mathsf X, d)$ is defined in \eqref{dtiltedinfo}. 
The following proposition implies that $\tilde{\mathcal V}(d) > \mathcal V(d)$ unless there is no noise. 
\begin{prop}
The noisy rate-dispersion function in \eqref{Vtilde} can be written as
\begin{align}
   \tilde {\mathcal V}(d)  = \Var{\jmath_{\mathsf X}(\mathsf X, d)}  +  \lambda^{\star 2} \Var{ \sd(\mathsf S, \mathsf Z^\star) | \sX, {\mathsf Z}^\star}. \label{Vnoisydecomp1}
\end{align}
\end{prop}
\begin{proof}
Note that 
the expression in \eqref{Vnoisydecomp1} is just an equivalent way to write the decomposition  \eqref{Vnoisydecomp}.  
To verify that the decomposition \eqref{Vnoisydecomp} indeed holds, write, using \eqref{Vtilde} and \eqref{dtiltedinfo}, 
\begin{align}
\tilde{\mathcal V}(d)
 &= \Var{\jmath_{\mathsf X}(\mathsf X, d) 
 + \lambda^\star \sd(\mathsf S, \mathsf Z) 
  - \lambda^\star \E{\mathsf d(\mathsf S, \mathsf Z^\star) | \mathsf X, \mathsf Z^\star} }\\
 &= \Var{\jmath_{\mathsf X}(\mathsf X, d)} 
 + \lambda^{\star 2} \Var{ \sd(\mathsf S, \mathsf Z^\star) | \mathsf X, \mathsf Z^\star}
  \notag\\
 &+ 2 \lambda^{\star} \mathrm{Cov} \left( \jmath_{\mathsf X}(\mathsf X, d), \sd(\mathsf S, \mathsf Z) - \E{\mathsf d(\mathsf S, \mathsf Z) | \mathsf X, \mathsf Z} \right), 
\end{align}
where the covariance is zero: 
\begin{align}
&~ \E{ \left( \jmath_{\mathsf X}(\mathsf X, d) - R(d)\right) \left(\sd(\mathsf S, \mathsf Z) - \E{\mathsf d(\mathsf S, \mathsf Z) | \mathsf X, \mathsf Z} \right) }
\notag \\
= &~ 
\E{ \left( \jmath_{\mathsf X}(\mathsf X, d) - R(d)\right) \E{ \left(\sd(\mathsf S, \mathsf Z) - \E{\mathsf d(\mathsf S, \mathsf Z) | \mathsf X, \mathsf Z} \right) |  \mathsf X, \mathsf Z }  }\\
= &~ 0.
\end{align}

\end{proof}

\begin{remark}
Generalizing the observation made by Ingber and Kochman \cite{ingber2011dispersion} in the noiseless lossy compression setting, we note using Theorem \ref{thm:nsc:ingber} that the noisy rate-dispersion function admits the following representation:
\begin{equation}
\tilde{{\mathcal V}} (d) =  \Var{ \dot{\mathbb R}_{\mathsf S, \mathsf X}(\mathsf S, \mathsf X, \mathsf Z^\star, d)} 
\end{equation}
where $\dot {\mathbb R}_{\mathsf S, \mathsf X}$ is defined in \eqref{eq:fdot}. 
\end{remark}

\section{Example: erased fair coin flips}
\label{sec:BES}

\subsection{Erased fair coin flips}
Let a binary equiprobable source be observed by the encoder through a binary erasure channel with erasure rate $\delta$. The goal is to minimize the bit error rate with respect to the source.  For $\frac \delta 2 \leq d \leq \frac 1 2$, the rate-distortion function is given by
\begin{equation}
R(d) = (1-\delta)\left(\log 2 - h\left( \frac{d - \frac \delta 2}{1 - \delta}\right)\right) \label{Rdbes},
\end{equation}
where $h(\cdot)$ is the binary entropy function, and \eqref{Rdbes} is obtained by solving the optimization in \eqref{Rdnoisy} which is achieved by $P_{\mathsf Z^\star}(0) = P_{\mathsf Z^\star}(1) = \frac 1 2$ and 
\begin{equation}
 P_{\mathsf X|\mathsf Z^\star}(a|b) = 
\begin{cases}
  1 - d - \frac \delta 2 & b = a\\
  d - \frac \delta 2 & b \neq a \neq ?\\
  \delta	& a = ?
\end{cases}
\label{eq:PX|Zerased}
\end{equation}
where $a \in \{0, 1, ?\}$ and $b \in \{0, 1\}$, so 
\begin{align}
 \tilde{\jmath}_{\mathsf S,\mathsf X}( \mathsf S, \mathsf X, \mathsf Z^\star, d)
 =&~ \imath_{\mathsf X; \mathsf Z^\star}(\mathsf X; \mathsf Z^\star) + \lambda^\star \sd(\mathsf S, \mathsf Z^\star) - \lambda^\star d \label{jdensityBES}\\
 =&~ - \lambda^\star d + 
\begin{cases}
 \log \frac{2}{1 + \exp( -\lambda^\star)} & \text{w.p. } 1 - \delta \\
\lambda^\star & \text{w.p. } \frac \delta 2\\
0  & \text{w.p. } \frac \delta 2 \label{jBES}
\end{cases}
\end{align}
The rate-dispersion function is given by the variance of \eqref{jBES}:
\begin{align}
\label{dispersionBES}
\tilde {\mathcal V}(d) &= \delta(1-\delta) \log^2 \cosh \left( Ê\frac{\lambda^\star}{2 \log e} \right)  + \frac \delta 4 \lambda^{\star 2}, \\
\lambda^\star &= - R^\prime(d) = \log \frac{1 - \frac \delta 2 - d}{d - \frac \delta 2 } \label{eq:slopeerased}.
\end{align}

 A tight converse result can be obtained by particularizing the bound in Corollary \ref{cor:Cnoisya} with $P_{\bar X^k | \bar Z^k}$ chosen to be a product of $k$ independent copies of \eqref{eq:PX|Zerased} and $\lambda$ chosen to be $k$ times $\lambda^\star$ in \eqref{eq:slopeerased}. As detailed in Remark \ref{rem:symmetry}, due to the symmetry of the erased coin flips setting such a choice would eliminate the need to perform the optimization over $z$ in the right side of \eqref{Cnoisya}, yielding a particularly easy-to-compute bound. 
In \cite{kostina2011fixed}, we showed an even tighter converse result for the erased coin flips setting:  
\begin{thm}[{Converse, BES \cite[Theorem 32]{kostina2011fixed}}]
\label{thm:CBES}
In the erased fair coin flips setting, 
\begin{align}
\epsilon \geq& \sum_{j = 0}^{k} {k \choose j} \delta^{j}(1-\delta)^{k-j}
\notag\\ \cdot&
 \sum_{i = 0}^{j} 2^{-j}{j \choose i} \left[ 1- M^\star(k, d, \epsilon) 2^{-(k-j)}\binosum{k-j}{\lfloor k d - i\rfloor }\right]^+,
\label{CBES}
\end{align}
where 
\begin{equation}
\binosum{k}{j} 
\triangleq \sum_{i = 0}^{j} {k \choose i} \label{binosum},
\end{equation}
with the convention $\binosum{k}{j} = 0$ if $j < 0$ and $\binosum{k}{j} = \binosum{k}{k}$ if $j > k$.
\end{thm}

The following achievability bound from \cite{kostina2011fixed} turns out to be just a particularization of Theorem \ref{thm:Anoisy} to the erased coin flips setting with $P_{\bar Z}$ chosen to be equiprobable on the set of $k$-bit binary strings. 
\begin{thm}[{Achievability, BES \cite[Theorem 33]{kostina2011fixed}}]
\label{thm:ABES}
In the erased fair coin flips setting, 
\begin{align}
\epsilon \leq& \sum_{j = 0}^k {k \choose j} \delta^{j}(1-\delta)^{k-j} 
\notag\\
\cdot& 
\sum_{i = 0}^{j}2^{-j} {j \choose i}\left( 1 - 2^{-(k-j)}\binosum{k-j}{\lfloor kd - i \rfloor} \right)^{M^\star(k, d, \epsilon)}.
\label{ABES}
\end{align}
\end{thm}

\subsection{Erased fair coin flips: surrogate rate-distortion problem}
\label{sec:BESeq}
According to \eqref{dbar}, the distortion measure of the surrogate rate-distortion problem is given by 
\begin{align}
 \bar {\mathsf d}(1, 1) = \bar {\mathsf d}(0, 0) &= 0 \label{dBESeq1},\\
 \bar {\mathsf d}(1, 0) = \bar {\mathsf d}(0, 1) &= 1,\\
 \bar {\mathsf d}(?, 1) = \bar {\mathsf d}(?, 0) &= \frac 1 2  \label{dBESeq2}.
\end{align}
The $\mathsf d$-tilted information is given by taking the expectation of \eqref{jBES} with respect to $\mathsf S$:
\begin{align}
 \jmath_{\mathsf X}(\mathsf X, d)
 =&~ - \lambda^\star d + 
\begin{cases}
 \log \frac{2}{1 + \exp( -\lambda^\star)} & \text{w.p. } 1 - \delta \\
\frac{\lambda^\star}{2} & \text{w.p. }  \delta  
\end{cases}
\label{jBESeq}
\end{align}
Its variance is equal to
\begin{align}
\label{dispersionBESeq}
 {\mathcal V}(d) &= \delta(1-\delta) \log^2  \cosh \left( Ê\frac{\lambda^\star}{2 \log e} \right) \\
 &= \tilde {\mathcal V}(d) - \frac \delta 4 \lambda^{\star 2}. 
\end{align}

Achievability and converse bounds for the ternary source with binary representation alphabet and the distortion measure in \eqref{dBESeq1}--\eqref{dBESeq2} are obtained as follows. 

\begin{thm}[Converse, surrogate BES]
\label{thm:CBESeq}
Any $(k, M,d,\epsilon)$ code must satisfy
\begin{align}
\epsilon \geq& \sum_{j = 0}^{\lfloor 2 kd \rfloor} {k \choose j} \delta^{j}(1-\delta)^{k-j} \left[ 1- M 2^{-(k-j)}\binosum{k-j}{\lfloor k d - \frac 1 2 j\rfloor }\right]^+.
\label{CBESeq}
\end{align}
\end{thm}
\begin{proof}
 Fix a $(k, M, d, \epsilon)$ code. 
While $j$ erased bits contribute $ \frac j {2k}$ to the total distortion regardless of the code, the probability that $k-j$ nonerased bits lie within Hamming distance $\ell$ of their representation can be upper bounded using \cite[Theorem 15]{kostina2011fixed}:
\begin{align}
&~
\Prob { (k-j) \bar{\mathsf d}(X^{k-j}, Z^{k-j}) \leq \ell  \mid \text{ no erasures in }X^{k-j} }
\notag\\
\leq&~ M2^{-k+j} \binosum{k-j}{\ell} \label{eq:cbeq1}.
\end{align}
Using \eqref{eq:cbeq1}, the probability that the total distortion is within $d$ is written as 
\begin{align}
&~
\mathbb P \left[ \bar {\mathsf d}(X^k,Z^k )\leq d \right] 
\notag\\
=&~
\sum_{j = 0}^{\lfloor 2 kd \rfloor} \mathbb P[ j \text{ erasures in } X^k] 
\notag \\
\cdot&~
 \mathbb P \left[ (k-j) \bar {\mathsf d}(X^{k-j}, Z^{k-j}) \leq kd - \frac 1 2 j |  \text{ no erasures in }X^{k-j}\right] 
\notag
 \\
\leq &~ \sum_{j = 0}^{\lfloor 2 kd \rfloor} {k \choose j} \delta^{j}(1-\delta)^{k-j} \min\left\{ 1, \ M 2^{-(k-j)} \binosum{k-j}{\lfloor kd - \frac 1 2 j \rfloor }\right\}. 
\end{align}
\end{proof}

\begin{thm}[Achievability, surrogate BES]
\label{thm:ABESeq}
There exists a $(k, M,d,\epsilon)$ code such that
\begin{align}
\epsilon \leq& \sum_{j = 0}^{k} {k \choose j} \delta^{j}(1-\delta)^{k-j} \left( 1 - 2^{-(k-j)}\binosum{k-j}{\lfloor kd - \frac 1 2 j \rfloor} \right)^M.
\label{ABESeq}
\end{align}
\end{thm}
\begin{proof}
Consider the ensemble of codes with $M$ codewords drawn i.i.d. from the equiprobable distribution on $\{0,1\}^k$. Every erased symbol contributes $\frac 1 {2k}$ to the total distortion. The probability that the Hamming distance between the nonerased symbols and their representation exceeds $
\ell$, averaged over the code ensemble is found as in \cite[Theorem 16]{kostina2011fixed}:
\begin{align}
&~
\mathbb P \left[ (k-j)\bar {\mathsf d}(X^{k-j}, \mathsf C(\mathsf f(X^{k-j}))) > \ell | \text{ no erasures in }X^{k-j}\right] 
\notag\\
=&~ \left( 1 - 2^{-(k-j)}\binosum{k-j}{\ell}\right)^M,
\end{align}
where $\mathsf C(m)$, $m = 1, \ldots, M$ are i.i.d on $\{0,1\}^{k-j}$.
Averaging over the erasure channel, we have
\begin{align}
&~\mathbb P \left[ \mathsf d(S^k, \mathsf C(\mathsf f(X^k)))) > d \right]\notag\\
=&~\sum_{j = 0}^{k} \mathbb P[ j \text{ erasures in } X^k] 
\notag\\ \cdot &~
\mathbb P \left[ (k-j) \mathsf d(X^{k-j}, \mathsf C(\mathsf f(X^{k-j}))) > kd - \frac 1 2 j \right]  
\notag
\\
= &~ \sum_{j = 0}^{k} {k \choose j} \delta^{j}(1-\delta)^{k-j} \left( 1 - 2^{-(k-j)}\binosum{k-j}{\lfloor kd - \frac 1 2 j\rfloor}  \right)^M.
\end{align}
Since there must exist at least one code whose excess-distortion probability is no larger than the average over the ensemble, there exists a code satisfying \eqref{ABESeq}.
\end{proof}

The bounds in \cite[Theorem 32]{kostina2011fixed}, \cite[Theorem 33]{kostina2011fixed} and the approximation in Theorem \ref{thm:2ordernoisy} (with the remainder term equal to\footnote{It was shown in \cite[Theorem 34]{kostina2011fixed} that the $\bigo{\frac{\log k}{k}}$ term for the erased coin flips lies between $0$ and $\frac {\log k}{2k}$. The same bound on the remainder term can be shown for the surrogate noiseless problem, via an analysis of \eqref{CBESeq} and \eqref{ABESeq}. } $0$ and $\frac {\log k}{2k}$), as well as the bounds in Theorems \ref{thm:CBESeq} and \ref{thm:ABESeq} for the surrogate rate-distortion problem together with their Gaussian approximation, are plotted in Fig. \ref{fig:BES}. 

Note that 
despite the fact that the asymptotically achievable rate in both problems is the same, there is a very noticeable gap between their nonasymptotically achievable rates in the displayed region of blocklengths. For example, at blocklength 1000, the penalty over the rate-distortion function is $9 \%$ for erased coin flips and only $4\%$ for the surrogate source coding problem.


\begin{figure}[htp]
\begin{center}
    \epsfig{file=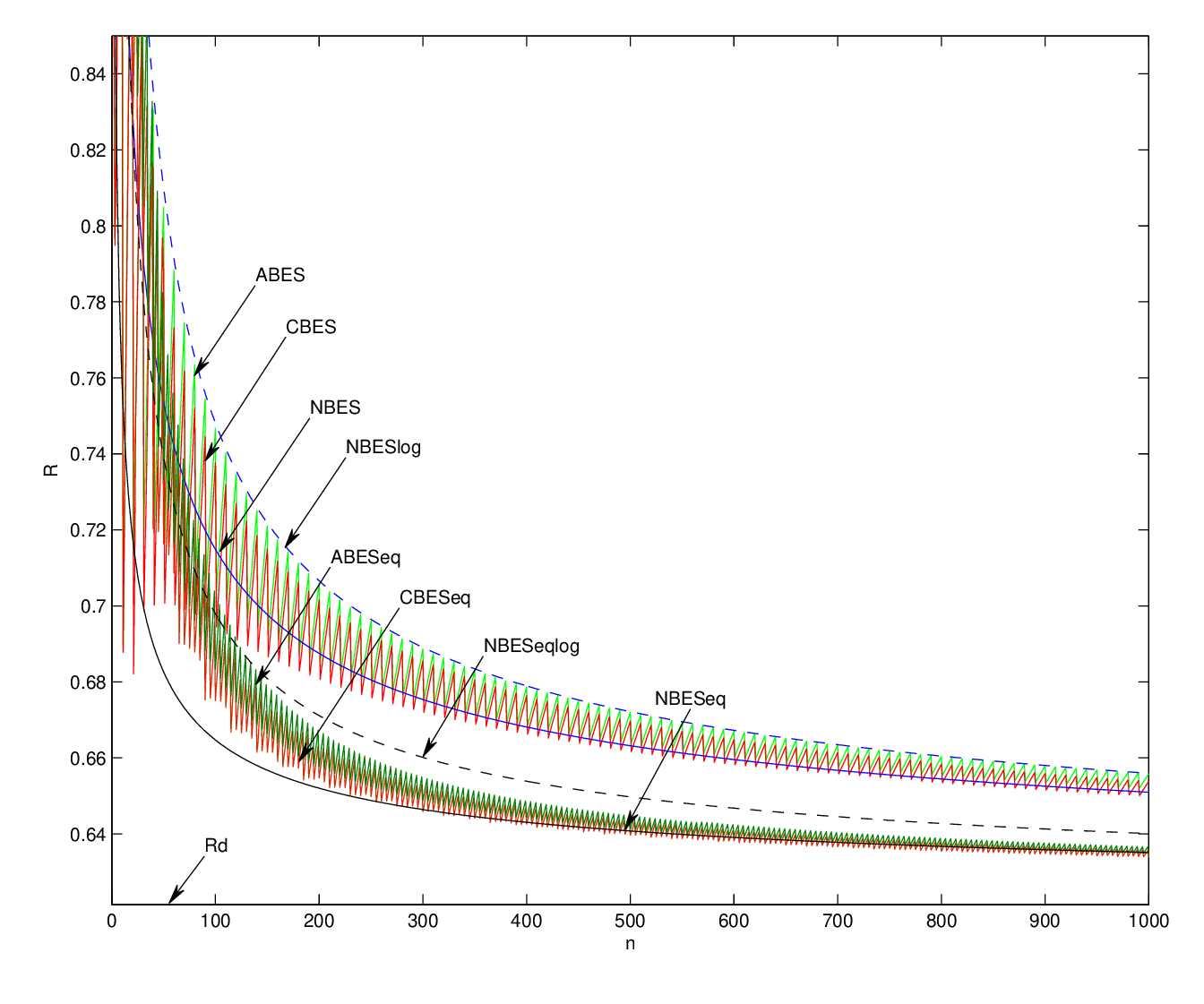,width=1\linewidth }
\end{center}
 \caption[Fair BMS observed through a BEC: rate-blocklength tradeoff. ]{Rate-blocklength tradeoff for the fair coin source observed through an erasure channel, as well as that for the surrogate problem, with $\delta = 0.1$, $d = 0.1$, $\epsilon = 0.1$. } \label{fig:BES}
\end{figure}


\section{Acknowledgement}
We thank the anonymous reviewer for the remarkably thorough review,
which is reflected in the final version.

\appendices
\section{Auxiliary results}
\label{appx:aux}

In this appendix, we state auxiliary results that are instrumental in the proof of Theorem \ref{thm:2ordernoisy}. 

\begin{thm}[{Berry-Esse\'en CLT, e.g. \cite[Ch. XVI.5 Theorem 2]{feller1971introduction}}]
\label{thm:BerryEsseen}
Fix a positive integer $k$. Let $W_i$, $i = 1, \ldots, k$ be independent. Then, for any real $t$
\begin{equation}
\left| \mathbb P \left[ \sum_{i = 1}^k W_i > k \left( \mu_k + t \sqrt {\frac{V_k}{ k}}\right) \right]  - Q(t) \right| \leq \frac {B_k}{\sqrt k},
\label{BerryEsseen}
\end{equation}
where
\begin{align}
\mu_k &= \frac 1 k \sum_{i = 1}^k \E{ W_i}  \label{DkBerryEsseen},\\
V_k &= \frac 1 k \sum_{i = 1}^k \Var{W_i} \label{VkBerryEsseen},\\
T_k &= \frac 1 k \sum_{i = 1}^k \E{ |W_i - \mu_i|^3 } \label{TkBerryEsseen},\\
B_k &= \frac{c_0 T_k}{V_k^{3/2}} \label{BkBerryEsseen},
\end{align}
and $0.4097 \leq c_0 \leq 0.5600$ ($0.4097 \leq c_0 < 0.4784$ for identically distributed $W_i$) \cite{shevtsova2010improvement,shevtsova2011absolute}. 
\end{thm}

The second result deals with minimization of the cdf of a sum of independent random variables.

Let $\mathcal D$ is a metric space with the metric $d \colon \mathcal D^2 \mapsto \mathbb R^+$. Define the random variable $Z$ on $\mathcal D$. Let $W_i, i = 1, \ldots, k$  be independent conditioned on $Z$. Denote
\begin{align}
\mu_k(z) &= \frac 1 k \sum_{i = 1}^k \E{ W_i|Z = z} \label{Dnz}, \\
V_k(z) &= \frac 1 k \sum_{i = 1}^k \Var{W_i| Z = z} \label{Vnz}, \\
T_k(z) &= \frac 1 k \sum_{i = 1}^k \E{ |W_i - \E{W_i} |^3 | Z = z } \label{Tnz}. 
\end{align}

Let $\ell_1$, $\ell_2$, $L_1$, $F_1$, $F_2$, $V_{\min}$ and $T_{\max}$ be positive constants. We assume that there exist $z^\star \in \mathcal D$ and sequences $\mu_k^\star$, $V_k^\star$ such that for all $z \in \mathcal D$,
\begin{align}
  \mu^\star_k - \mu_k(z) 
 &\geq 
 \ell_1 d^2\left(z, z^\star \right) - \frac{\ell_2}{k} \label{ell_1tighter},\\
 \mu^\star_k - \mu_k(z^\star) 
 &\leq 
  \frac{L_1}{k} \label{L_1} ,  \\
\left | V_k(z)  - V_k^\star \right | &\leq F_1 d \left( z, z^\star \right) + \frac{F_2}{k} \label{F_1tighter},\\
V_{\min} &\leq V_k(z)  \label{V_min},\\
 T_k(z) &\leq T_{\max} \label{T_max}.
\end{align}

\begin{thm}[{\hspace{-.1mm}\cite[Theorem A.6.4]{kostina2013thesis}}]
In the setup described above, under assumptions \eqref{ell_1tighter}--\eqref{T_max}, for any $A> 0$, there exists a $K \geq 0$ such that, for all $\left|\Delta\right| \leq A 2 \ell_1 T_{\max}^{\frac 1 3} V_{\min}^{\frac 5 2}F_1^{-2}$ and all  sufficiently large $k$: 
\begin{align}
 \min_{z \in \mathcal D} \Prob{\sum_{i = 1}^k W_i \leq  k\left(\mu^\star_k - \Delta\right)| Z = z} 
 &\geq   
 Q\left( \Delta \sqrt{\frac{k} { V^\star_k} }\right)
 -  \frac {K} {\sqrt k} \label{jscc:minprob:4}.
\end{align}

 \label{thm:minprob}
 \end{thm}

The following two theorems summarize some crucial properties of $\mathsf d$-tilted information. It is convenient to introduce the notation 
\begin{equation}
\mathbb R_{X}(d) \triangleq  \mathbb R_{X,X}(d), 
\end{equation}
which is just the usual noiseless rate-distortion function. 

\begin{thm}[{\hspace{-.1mm}\cite[Theorem 2.1]{kostina2013thesis}}]
Fix $d > d_{\min}$. For $P_{Z}^\star$-almost every $z$, it holds that
\begin{align}
 \jmath_{X}(x, d) &=  \imath_{X; Z^\star}(x; z)  + \lambda^\star {\mathsf d}(x, z) - \lambda^\star d  \label{jddensity},
\end{align}
where $\lambda^\star = -\mathbb R_{ X}^\prime(d)$, and  $P_{XZ^\star} = P_{X}P_{Z^\star|X}$. Moreover,
\begin{align}
\mathbb R_{X}(d) &= \min_{P_{Z|X}} \E{ \imath_{X; Z}(X; Z)  + \lambda^\star {\mathsf d}(X, Z)  }    - \lambda^\star d\label{R(d)minj}\\
&= \min_{P_{Z|X}} \E{ \imath_{X; Z^\star}(X; Z)  + \lambda^\star {\mathsf d}(X, Z)  } - \lambda^\star d  \label{R(d)minjstar}\\
&= \E{\jmath_{X}(X, d)}\label{Ejd},
\end{align}
and for all $z \in \widehat{\mathcal M}$
\begin{equation}
\E{ \exp\left\{ \lambda^\star d - \lambda^\star { {\mathsf d}}(X, z) + \jmath_X(X, d)\right\} } \leq 1 \label{csiszar},
\end{equation}
with equality for $P_{Z}^\star$-almost every $z$. 
\label{thm:csiszar}
\end{thm}

The next result establishes a relationship between partial derivatives of $\mathbb R_X(d)$ with respect to coordinates of source distribution  and the $\sd$-tilted information. To properly define differentiation on the simplex, we consider $ \mathbb R_{X}(d) = \mathbb R_{P_X}(d)$ as a function of $\mathcal M$-dimensional probability vector $P_X$. We extend the definition of $\mathbb R_{P_X}(d)$ to $\mathcal M$-dimensional nonnegative vectors $Q_X = \{Q_X(x)\}_{x \in \mathcal M}$ as follows: 
\begin{equation}
 \mathbb R_{Q_X}(d) \triangleq \mathbb R_{P_{\bar X}}(d),
\end{equation}
where $P_{\bar X}(x) = \frac{Q_X(x)}{\sum_{x \in \mathcal M} Q_X(x)}$, so that $P_{\bar X}$ is a probability vector. 
For $a \in \mathcal A$, denote the partial derivatives 
\begin{equation}
\dot{\mathbb R}_X(a, d) =  \left. \frac{\partial }{\partial Q_{X}(a)} \mathbb R_{Q_X}(d) \right |_{Q_{X} = P_X}.
\end{equation}

\begin{thm}[{\hspace{-.1mm}\cite[Theorem 2.2]{kostina2013thesis}}]
Assume that $\mathcal X$ is a finite set. 
$\mathbb R_{\bar X}(d) = I(\bar X; \bar Z^\star)$. Then, 
\begin{align}
\dot{\mathbb R}_X(a, d)
 &= \jmath_X(a, d) - \mathbb R_X(d) \label{dif2},\\
 \Var{ \dot{\mathbb R}_X(X, d)} &= \Var{\jmath_X(X, d)}.  \label{ingber}
\end{align}
\label{thm:ingber}
\end{thm}
\begin{remark}
The equality in \eqref{ingber} (with a slightly different understanding of differentiation with respect to a probability vector coordinate) was first observed in \cite{ingber2008distortion}. 
\end{remark}

\apxonly{
\begin{proof}
Since by the assumption \eqref{jddensity} particularized to $P_{\bar X}$ holds for $P_{Z^\star}$-almost every $z$, we may write
\begin{align}
\E{\jmath_{\bar X}(X, d)} &= \E{ \imath_{\bar X; \bar Z^\star}(X; Z^\star) } - \mathbb R^\prime_{\bar X}(d) \E{ \mathsf d(X, Z^\star) - d}\\
&= \E{ \imath_{\bar X; \bar Z^\star}(X; Z^\star) }
\end{align}
Therefore (in nats)
\begin{align}
&~ \left. \frac{\partial }{\partial P_{\bar X}(a)} \E{ \jmath_{\bar X}(X, d)} \right|_{P_{\bar X} = P_X} \\
=&~ 
\left. \frac{\partial }{\partial P_{\bar X}(a)} \E{ \log P_{\bar X| \bar Z^\star}(X; Z^\star)}  \right|_{P_{\bar X} = P_X} 
-\left. \frac{\partial }{\partial P_{\bar X}(a)} \E{ \log P_{\bar X}(X)}  \right|_{P_{\bar X} = P_X} \\
=&~   \left.\frac{\partial }{\partial P_{\bar X}(a)}  \E{ \frac{P_{\bar X| \bar Z^\star}(X; Z^\star)}{P_{ X|  Z^\star}(X; Z^\star)}}  \right|_{P_{\bar X} = P_X} 
- 
\left. \E{\frac 1 {P_X(X)} \frac{\partial }{\partial P_{\bar X}(a)} P_{\bar X}(X)}  \right|_{P_{\bar X} = P_X} \\
=&~ \left. \frac{\partial }{\partial P_{\bar X}(a)} 1 \right|_{P_{\bar X} = P_X} - 1\\
=&~ - 1
\end{align}
This proves \eqref{dif1}. To show \eqref{dif2}, we invoke \eqref{dif1} to write
\begin{align}
\dot{\mathbb R}_X(a) &=  \frac{\partial }{\partial P_{\bar X}(a)}\E{ \jmath_{\bar X}(\bar X, d)}  \mid_{P_{\bar X} = P_X}\\
 &= \jmath_X(a, d) +  \frac{\partial }{\partial P_{\bar X}(a)}\E{ \jmath_{\bar X} ( X, d)} \mid_{P_{\bar X} = P_X} \\
 &= \jmath_X(a, d) - \log e
\end{align}
Finally, \eqref{ingber} is an immediate corollary to \eqref{dif2}. 
\end{proof}
}

%

\section{Proof of Theorem \ref{thm:nsc:ingber}}
\label{appx:ingber}
Denote for brevity $\bar \lambda = \lambda(P_{\bar S \bar X})$. Theorem \ref{thm:nsc:ingber} is obtained using the chain rule of differentiation, as shown in~\eqref{eq:daa}--\eqref{eq:daaa}~at the bottom of the next page.  
\begin{table*}[b]
\normalsize
\vspace*{4pt}
\hrulefill
\begin{align}
&~\left. \frac{\partial }{\partial Q_{ S  X  Z}(s, x, z) }  \E{ \imath_{\bar X; \bar Z}(\bar X; \bar Z)}  + \bar \lambda \E{\sd (\bar S, \bar Z)} - \bar \lambda d \right|_{Q_{ S  X  Z^\star} = P_{S X Z^\star} } \notag\\
=&~ \tilde \jmath_{S, X}(s, x, z, d)  - \mathbb R_{S, X}(d) 
+ \left. \frac{\partial }{\partial Q_{ S  X  Z}(s, x, z) }  \left(   \E{ \imath_{\bar X; \bar Z}( X;  Z^\star)}  + \bar \lambda\E{ \mathsf d( S, Z^\star)} - \bar \lambda d \right) \right|_{Q_{ S  X  Z} = P_{S X Z^\star} } 
\label{eq:daa}\\
=&~  \tilde \jmath_{S, X}(s, x, z, d) - \mathbb R_{S, X}(d)  + \left. \frac{\partial }{\partial Q_{ S  X  Z^\star}(s, x, z)}   \E{ \imath_{\bar X; \bar Z}( X;  Z^\star)} \right|_{P_{\bar S \bar X \bar Z} = P_{S X Z^\star} }
\end{align}

\begin{align}
\left. \frac{\partial }{\partial Q_{ S  X  Z}(s, x, z)} \E{ \imath_{\bar X; \bar Z}( X;  Z^\star)} \right |_{P_{\bar S \bar X \bar Z^\star } = P_{S X Z^\star}} 
=
&~ 
\left. \frac{\partial }{\partial Q_{ S  X  Z}(s, x, z)} \E{ \log P_{\bar Z | \bar X}(Z^\star | X)} 
 - \E{ \log P_{\bar Z}(Z^\star)}  \right|_{P_{\bar S \bar X \bar Z^\star} = P_{S X Z^\star}} \\
=&~   \left. \frac{\partial }{\partial Q_{ S  X  Z}(s, x, z, d)}  \E{ \frac{P_{\bar Z| \bar X}(Z^\star| X)}{P_{ Z^\star |X}(Z^\star | X)}}   - \E{\frac {P_{\bar X} (\bar X)} {P_X(X)} } \right|_{P_{\bar S \bar X \bar Z^\star} = P_{S X Z}} \\
=&~ 0
\label{eq:daaa}
\end{align}

\end{table*}

\section{Proof of the converse part of Theorem \ref{thm:2ordernoisy}}
\label{appx:2ordernoisyC}


The following concentration is instrumental in the proof of the converse side of the asymptotic Gaussian approximation. 
\begin{lemma}
Let $X_1, \ldots, X_k$ be independent on $\mathcal A$ and distributed according to $P_{\mathsf X}$. For all $k$ and all $\gamma > 0$, it holds that 
\begin{equation}
 \Prob{ \left| \type{X^k} - P_{\mathsf X} \right|^2 > \gamma } \leq 2 |\mathcal A| \exp\left( - \frac k 2 \frac \gamma {|\mathcal A|}\right) \label{eq:yu},
\end{equation} 
where $\left| \cdot \right|$ denotes the Euclidean norm of its $|\mathcal A|$-dimensional vector argument. 
\label{lemma:yu} 
\end{lemma}
\begin{proof}
Observe first that
\begin{align}
&~ \left\{P_{\bar {\mathsf X}} \in \mathcal P \colon \left| P_{\bar {\mathsf X}}  - P_{\mathsf X} \right|^2 > \gamma \right\} \notag\\
&~\subseteq \left\{ P_{\bar {\mathsf X}} \in \mathcal P \colon \exists a \in \mathcal A \colon | P_{\bar {\mathsf X}} (a) - P_{\mathsf X}(a)|^2 > \frac {\gamma}{|\mathcal A|} \right\},
\end{align}
where $\mathcal P$ is the set of all distributions on $\mathcal A$. 
For all $a \in \mathcal A$, it holds by Hoeffding's inequality (cf. Yu and Speed \cite[(2.10)]{yu1993rateofconvergence}) that
\begin{align}
&~ \Prob{ \type{X^k} =P_{\bar {\mathsf X}} \colon  | P_{\bar {\mathsf X}} (a) - P_{\mathsf X}(a)|^2 > \frac {\gamma}{|\mathcal A|}} \notag\\
\leq&~ 2 \exp \left( - \frac k 2 \frac \gamma {|\mathcal A|}\right), 
\end{align}
and \eqref{eq:yu} follows by the union bound. 
\end{proof}

Let $P_{\mathsf Z | \mathsf X} \colon \mathcal A \mapsto \hat {\mathcal S}$ be a stochastic matrix whose entries are multiples of $\frac 1 k$.  We say that the conditional type of $z^k$ given $x^k$ is equal to $P_{\mathsf Z | \mathsf X}$, $\type{z^k | x^k} = P_{\mathsf Z | \mathsf X}$, if the number of $a$'s in $x^k$ that are mapped to $b$ in $z^k$ is equal to the number of $a$'s in $x^k$ times $P_{\mathsf Z | \mathsf X}(b | a)$, for all $(a, b) \in \mathcal A \times \hat {\mathcal S}$.

 Let 
\begin{align}
\log M &= k R(d) + \sqrt {k \tilde{\mathcal V}(d)} \Qinv{\epsilon_k} -  \frac 1 2 \log k -  \log |\mathcal P_{[k]}| \notag\\
&- c|\mathcal A| \log k \label{logMC},
\end{align}
where $\epsilon_k = \epsilon + \bigo{\frac{\log k}{\sqrt k}}$ and constant $c$ will be specified in the sequel, and $\mathcal P_{[k]}$ denotes the set of all conditional $k-$types $\hat {\mathcal S} \to \mathcal A$. Note that $ |\mathcal P_{[k]}| \leq (k+1)^{|\mathcal A| |\mathcal S|}$.

 We weaken the bound in \eqref{Cnoisy} by choosing
\begin{align}
 P_{\bar X^k | \bar Z^k = z^k}(x^k) &= \frac 1 { | \mathcal P_{[k]} | } \sum_{P_{\mathsf X | \mathsf Z} \in \mathcal P_{[k]}} \prod_{i = 1}^k P_{\mathsf X | \mathsf Z = z_i}(x_i) 
   \label{PbarX},\\
  \lambda &= k \lambda(x^k) = k \mathbb R^\prime_{\type{x^k}}(d),\\
  \gamma &= \frac 1 2 \log k \label{Cnoisygamma}.
\end{align}  

By virtue of Theorem \ref{thm:Cnoisy}, the excess distortion probability of all $(M, d, \epsilon^\prime)$ codes where $M$ is that in \eqref{logMC} must satisfy
\begin{align}
\epsilon^\prime &\geq \mathbb E \Big [ \min_{ 
z^k \in \hat {\mathcal S}^k }
 \mathbb P \Big[ \imath_{\bar X^k|\bar Z^k \| X^k}(X^k; z^k)  \notag\\
 &+k \lambda(X^k)  ( \mathsf d(S^k, z^k) -  d)
 \geq \log M + \gamma  | X^k \Big] \Big] -  \exp(-\gamma) \label{-noisya}.
\end{align}
It suffices to prove that the right side of \eqref{-noisya} is lower bounded by $\epsilon$ for $M$ defined in \eqref{logMC}.

For a given pair $(x^k, z^k)$, 
abbreviate 
\begin{align}
\type{x^k} &= P_{\bar{\mathsf X}},\\
\type{z^k | x^k} &= P_{\bar{\mathsf Z} | \bar {\mathsf X}},\\
\lambda(x^k) &= \lambda_{\bar{\mathsf X}}.
\end{align}
For each $x^k$ and $z^k$, denote the independent random variables
\begin{align}
W_i &\triangleq I(\bar{\mathsf X}; \bar{\mathsf Z} )  + \lambda_{\bar{\mathsf X}}   \left( 
 \sd(S_i, z_i) -  d \right), ~ i = 1, \ldots, n.
\end{align}
Since 
\begin{align}
\E{\sum_{i = 1}^k  \sd(S_i, z_i)}  &= k \E{  {\mathsf d} (\mathsf S, \bar{\mathsf Z} ) }, \\
\Var{\sum_{i = 1}^k  \sd(S_i, z_i)} &= k\, \Var{  {\mathsf d} (\mathsf S, \bar{\mathsf Z} )  | \bar{\mathsf X}, \bar {\mathsf Z}  },
\end{align}
where $P_{\mathsf S \bar{\mathsf X} \bar {\mathsf Z}} = P_{\mathsf S | \mathsf X} P_{\bar{\mathsf X} \bar {\mathsf Z}}$, 
in the notation of Theorem \ref{thm:minprob} (where $z = P_{\bar{\mathsf Z} | \bar{\mathsf X}}$ is now a conditional type) we have
\begin{align}
\mu_k(P_{ \bar{\mathsf Z} | \bar{\mathsf X} }) &=  I(\bar{\mathsf X}; \bar{\mathsf Z} )  + \lambda_{\bar{\mathsf X}} \left( \E{  {\mathsf d} (\mathsf S, \bar{\mathsf Z} ) }  - d \right),\\
V_k(P_{\bar{\mathsf Z} | \bar{\mathsf X}}) &= \lambda_{\bar{\mathsf X}}^2 \, \Var{  {\mathsf d} (\mathsf S, \bar{\mathsf Z} )  | \bar{\mathsf X}, \bar {\mathsf Z}  }, \\
T_k(P_{\bar{\mathsf Z} | \bar{\mathsf X}}) &= \lambda_{\bar{\mathsf X}}^3\, \E{ \left|  {\mathsf d} (\mathsf S,\bar{\mathsf Z} ) - \E{ {\mathsf d}(\mathsf S,  \bar{\mathsf Z}) | \bar {\mathsf X}, \bar{\mathsf Z}}\right|^3 }.
\end{align}

We define $P_{\bar{\mathsf X} | \bar {\mathsf Z} }$ through $P_{\bar {\mathsf X}}  P_{\bar{\mathsf Z} |\bar{\mathsf X}} $ and lower-bound the sum in \eqref{PbarX} by the term containing $P_{ \bar{\mathsf X} | \bar{\mathsf Z}}$, concluding that 
\begin{align}
 &~ \imath_{\bar X^k|\bar Z^k \| X^k}(x^k; z^k)  + k \lambda(x^k) ( \mathsf d(S^k, z^k) -  d) \notag\\
 \geq&~ 
  k I( \bar{\mathsf X}; \bar{\mathsf Z} ) + k D(\bar{\mathsf X} \| \mathsf X)
  \notag\\
    &+ \lambda_{\bar{\mathsf X}} \left( \sum_{i = 1}^k  {\mathsf d}(S_i , z_i) -  k d \right) - \log \left| \mathcal P_{[k]} \right|\\
  &= \sum_{i = 1}^k W_i + k D(\bar{\mathsf X} \| \mathsf X) - \log \left| \mathcal P_{[k]}\right|
  \\
&\geq
 \sum_{i = 1}^k W_i - \log \left| \mathcal P_{[k]}\right| \label{str0}.
\end{align}

Weakening \eqref{-noisya}  using \eqref{str0}, we write
\begin{align}
\epsilon^\prime &\geq  \E{ \min_{ 
z^k \in \hat {\mathcal S}^k }
 \Prob{ \sum_{i = 1}^k W_i \geq \log M + \gamma + \log \left| \mathcal P_{[k]}\right|
   | X^k} } \notag\\
   &-  \exp(-\gamma) \label{-noisyb}.
\end{align}
We identify the typical set of channel outputs:
\begin{equation}
 \mathcal T_{k} = \left\{x^k \in \mathcal A^k \colon 
  \left| 
\type{x^k} - P_{\mathsf X}
  \right|^2 \leq |\mathcal A| \frac{\log k}{k}  \right\} \label{typical},
\end{equation}
where $| \cdot |$ is the Euclidean norm. 
Noting that by Lemma \ref{lemma:yu},
\begin{equation}
\Prob{X^k \notin \mathcal T_k} \leq \frac{2 |\mathcal A|}{\sqrt k}, \label{str:typ}
\end{equation}
we proceed to evaluate the minimum in \eqref{-noisyb} for $x^k \in \mathcal T^k$. 

Denote the backward conditional distribution that achieves $\mathbb R_{\bar {\mathsf X}}(d)$ by $P_{\bar {\mathsf X} | \bar {\mathsf Z}^\star}$. 
Write
\begin{align}
 \mu_k(P_{ \bar{\mathsf Z} | \bar{\mathsf X} } ) &= I(\bar{\mathsf X}; \bar{\mathsf Z} )  + \lambda_{\bar{\mathsf X}} \left( \E{ \bar{\mathsf d}(\bar{\mathsf X}, \bar{\mathsf Z}) }  - d \right)\\
 &= \E{\imath_{\bar {\mathsf X}; \bar {\mathsf Z}^\star }(\bar {\mathsf X}; \bar {\mathsf Z}) + \lambda_{\bar{\mathsf X}} \bar{\mathsf d}(\bar{\mathsf X}, \bar{\mathsf Z}) } - \lambda_{\bar{\mathsf X}} d \notag\\
 &+ D\left( P_{\bar{\mathsf X} | \bar{\mathsf Z}} \| P_{\bar{\mathsf X} | \bar{\mathsf Z}^\star} | P_{\bar {\mathsf Z}}\right)\\
 &\geq \mathbb R_{\bar{\mathsf X}} (d) + D\left( P_{\bar{\mathsf X} | \bar{\mathsf Z}} \| P_{\bar{\mathsf X} | \bar{\mathsf Z}^\star} | P_{\bar {\mathsf Z}}\right) \label{-Cn1}\\
  &\geq \mathbb R_{\bar{\mathsf X}} (d) + \frac 1 2 \left| P_{\bar{\mathsf X} | \bar{\mathsf Z}} P_{\bar {\mathsf Z}} - P_{\bar{\mathsf X} | \bar{\mathsf Z}^\star} P_{\bar {\mathsf Z}} \right|^2 \log e  \label{-Cn2a},
\end{align}
where \eqref{-Cn1} is by Theorem \ref{thm:csiszar}, and \eqref{-Cn2a} is by Pinsker's inequality. If $V_k(P_{{\mathsf Z^\star} | {\mathsf X}}) = 0$, there is nothing to prove, as in that case the second term in \eqref{eq:Cnoisyasym} is almost surely $0$ for all $z \in \mathrm{supp}(P_{{\mathsf Z}^\star})$. In the sequel we assume $V_k(P_{{\mathsf Z^\star} | {\mathsf X}}) > 0$. 
Similar to the proof of \cite[(C.50)]{kostina2013thesis}, we conclude that the conditions of Theorem \ref{thm:minprob} are satisfied by $W_i$, $i = 1, \ldots, k$, with $\mu_k^\star = \mu_k(P_{ \bar{\mathsf Z}^\star | \bar{\mathsf X} })$ and $V_k^\star = V_k(P_{ \bar{\mathsf Z}^\star | \bar{\mathsf X} })$. 
Therefore, abbreviating
\begin{equation}
\Delta_k(P_{\bar {\mathsf X}}) \triangleq \log M + \gamma+ \log \left|\mathcal P_{[k]}\right| - k \mathbb R_{\bar{\mathsf X}}(d),
\end{equation}
where $M$ and $\gamma$ were chosen in \eqref{logMC} and  \eqref{Cnoisygamma},
we have
\begin{align}
 &~ \min_{ 
P_{\bar{\mathsf Z} | \bar{\mathsf X} }} \Prob{\sum_{i = 1}^k W_i \geq  \log M + \gamma  + \log \left| \mathcal P_{[k]}\right|
  | \mathrm{type}\left(X^k\right) = P_{\bar {\mathsf X} } } \notag\\
  \geq&~ Q\left( \frac {\Delta_k(P_{\bar{\mathsf X}})} {\lambda_{\bar{\mathsf X}}\sqrt{  k V_k(P_{ \bar{\mathsf Z}^\star | \bar{\mathsf X} })}} \right) - \frac K {\sqrt k},
  \label{str1}
\end{align}
where $K >0$ is that in Theorem \ref{thm:minprob}.  Noting that assumption \eqref{item:5} implies, via \eqref{nsc:ingber}, that $\bar{\mathsf d}_{\bar{\mathsf Z}^\star}(\mathsf s | \mathsf x) $ is continuously differentiable as a function of $P_{\bar {\mathsf X}}$ in a neighborhood of $P_{\mathsf X}$, we can apply a Taylor series expansion in the vicinity of $P_{\mathsf{X}}$ to $ \frac {1} {\lambda_{\bar{\mathsf X}}\sqrt{ V_k(P_{ \bar{\mathsf Z}^\star | \bar{\mathsf X} })}}$ to conclude that for some scalar $a$ and $K_1 > 0$
\begin{align}
&~Q\left( \frac {\Delta_k(P_{\bar{\mathsf X}})} {\lambda_{\bar{\mathsf X}}\sqrt{  k V_k(P_{ \bar{\mathsf Z}^\star | \bar{\mathsf X} })}} \right) \notag\\
\geq&~  
Q\left( \frac {\Delta_k(P_{\bar {\mathsf X}})} {\lambda_{{\mathsf X}}\sqrt{  k V_k(P_{ {\mathsf Z}^\star | {\mathsf X} })}} \left( 1 + a \sqrt{\frac{\log k}{k}}\right) \right)\\
\geq&~ 
Q\left( \frac {\Delta_k(P_{\bar {\mathsf X}})} {\lambda_{{\mathsf X}}\sqrt{  k V_k(P_{ {\mathsf Z}^\star | {\mathsf X} })}} \right) - K_1 {\frac{\log k}{\sqrt k}} \label{str2},
\end{align}
where to obtain \eqref{str2} we used
 \begin{equation}
 Q(x + \xi) \geq Q(x) - \frac{|\xi|^+}{\sqrt{2 \pi}} \label{Qlb},
\end{equation}
with $\xi \sim \frac{\log k}{\sqrt{k}}$, as $\Delta_k(P_{\bar{\mathsf X}}) = \bigo{\sqrt{k \log k}}$ for $x^k \in \mathcal T_k$.
On the other hand, by assumption \eqref{item:5} Taylor's theorem applies to $\mathbb R_{\bar{\mathsf X}}(d)$; thus, there exists $c > 0$ such that
\begin{align}
&~ \mathbb R_{\bar{\mathsf X}}(d) \notag\\
\geq&~ \mathbb R_{\mathsf X}(d) + \sum_{a \in \mathcal A} \left( P_{\bar{\mathsf X}}(a) - P_{\mathsf X}(a)\right) \dot{\mathbb R}_{\mathsf X}(a, d)
- c  \left| P_{\bar{\mathsf X}} - P_{\mathsf X}\right|^2 \\
=&~ \mathbb R_{\mathsf X}(d) 
+ \frac 1 k \sum_{i = 1}^k    \dot{\mathbb R}_{\mathsf X}(X_i, d)  - \E{ \dot{\mathbb R}_{\mathsf X}( \mathsf X, d) } 
- c  \left| P_{\bar{\mathsf X}} - P_{\mathsf X}\right|^2  \\
=&~  \E{    \jmath_{\mathsf X}(\bar{\mathsf X}, d)}  - c  \left| P_{\bar{\mathsf X}} - P_{\mathsf X}\right|^2 \label{-Cn4a}\\
\geq&~  \E{    \jmath_{\mathsf X}(\bar{\mathsf X}, d)}  - c |\mathcal A| \log k \label{-Cn4b},
\end{align}
where \eqref{-Cn4a} uses \eqref{dif2}, and \eqref{-Cn4b} is by the definition \eqref{typical} of the typical set of $x^k$'s. Therefore, introducing the random variable $G \sim \mathcal N(0, 1)$ which is independent of $X^k$, we may write
\begin{align}
 &~Q\left( \frac {\Delta_k(P_{\bar {\mathsf X}})} {\lambda_{{\mathsf X}}\sqrt{  k V_k(P_{ {\mathsf Z}^\star | {\mathsf X} })}} \right) \notag\\
 =&~\Prob{ k \mathbb R_{\bar {\mathsf X}}(d) + \lambda_{\mathsf X} \sqrt{k V_k\left( P_{{\mathsf Z}^\star | \mathsf X}\right) } G \geq \log M + \gamma + \log \left| \mathcal P_{[k]} \right|  }\\
 \geq&~ \Prob{ \E{    \jmath_{\mathsf X}(\bar{\mathsf X}, d)} + \lambda_{\mathsf X} \sqrt{k V_k\left( P_{{\mathsf Z}^\star | \mathsf X}\right) } G \geq \log M + a_k } \label{str4aa},
\end{align}
where $a_k = \bigo{\log k}$ is defined as
\begin{equation}
a_k \triangleq  \gamma + \log \left| \mathcal P_{[k]} \right| + c |\mathcal A| \log k.
\end{equation}

Finally, collecting \eqref{-noisya}, \eqref{str1}, \eqref{str2} and \eqref{str4aa}, we obtain \eqref{str4a}-\eqref{str6}, shown at the bottom of the next page, where

\begin{table*}[!b]
\normalsize
\vspace*{4pt}
\hrulefill
\begin{align}
\epsilon^\prime &\geq  \E{ \min_{ 
z^k \in \hat {\mathcal S}^k }
 \Prob{ \sum_{i = 1}^k W_i \geq \log M + \gamma + \log \left| \mathcal P_{[k]}\right|
   | X^k} \1{X^k \in \mathcal T_k }} 
   - \Prob{X^k \notin \mathcal T_k}  - \exp(-\gamma) \label{str4a}\\
   &\geq
   \Prob{ \sum_{i = 1}^k \jmath_{\mathsf X}(X_i, d) + \lambda_{\mathsf X} \sqrt{k V_k\left( P_{{\mathsf Z}^\star | \mathsf X}\right) } G \geq \log M + a_k, X^k \in \mathcal T_k  } 
   - \Prob{X^k \notin \mathcal T_k}  - \exp(-\gamma) - K_1 {\frac{\log k}{\sqrt k}} \\
   &\geq
      \Prob{ \sum_{i = 1}^k \jmath_{\mathsf X}(X_i, d) + \lambda_{\mathsf X} \sqrt{k V_k\left( P_{{\mathsf Z}^\star | \mathsf X}\right) } G \geq \log M + a_k  } 
   - 2 \Prob{X^k \notin \mathcal T_k}  - \exp(-\gamma)   - K_1 {\frac{\log k}{\sqrt k}}   \label{str4}\\
    &\geq \epsilon_k - \frac{B}{\sqrt k} -  2 \Prob{X^k \notin \mathcal T_k}  - \exp(-\gamma)  - K_1 {\frac{\log k}{\sqrt k}}    \label{str5}\\
    &\geq \epsilon_k - \frac{B + 4 |\mathcal A| + 1 + K_1 \log k}{\sqrt k} \label{str6}
\end{align}

\end{table*}

\begin{itemize}
 \item \eqref{str4} is by the union bound;
 \item  \eqref{str5} is by the Berry-Ess\'een theorem (Theorem \ref{thm:BerryEsseen}), the choice of $M$ in \eqref{logMC}, the observation \eqref{Vnoisydecomp1}, and $B$ is the Berry-Ess\'een ratio;
 \item \eqref{str6} substitutes \eqref{Cnoisygamma} and \eqref{str:typ}.
\end{itemize}
Now, $\epsilon^\prime \geq \epsilon$ follows by letting  in \eqref{logMC}
\begin{equation}
\epsilon_k =  \epsilon +   \frac{B + 4 |\mathcal A| + 1 + K_1 \log k}{\sqrt k}.
\end{equation}

\section{Proof of the achievability part of Theorem \ref{thm:2ordernoisy}}
\label{appx:2ordernoisyA}

The proof consists of an analysis of the random code described in Theorem \ref{thm:Anoisy} with $M$ codewords drawn from the distribution $P_{Z^{k \star}} = P_{\mathsf Z^\star} \times \ldots \times P_{\mathsf Z^\star}$, where $\mathsf Z^\star$ achieves the rate-distortion function $\mathbb R_{\mathsf S, \mathsf X}(d)$. Theorem \ref{thm:Anoisy1} provides a means to study the performance of that code ensemble. Namely, it implies that the averaged over the ensemble excess-distortion probability $\epsilon^\prime$  is bounded by
\begin{align}
\epsilon^\prime &\leq \Prob{  g_{Z^{k \star}}(X^k, U) > \log \gamma } + e^{- \frac M \gamma} \\
&\leq  \Prob{  g_{Z^{k \star}}(X^k, U) > \log \gamma,  U \in I_k, X^k \in \mathcal T_k} \notag\\
&+ \frac {2+ 2 |\mathcal A|} {\sqrt k} + e^{- \frac M \gamma} \label{Anoisyweaken},
\end{align}
where
\begin{align}
I_k \triangleq \left[\frac 1 {\sqrt k}, 1 - \frac 1 {\sqrt k}\right], 
\end{align}
and $\mathcal T_k$ is the typical set of $X^k$ defined in \eqref{typical} so that \eqref{str:typ} holds, and \eqref{Anoisyweaken} is an obvious weakening. 
We let  
\begin{align}
 \log M &= \log \gamma + \log \log_e \sqrt k, \\
 \log \gamma &= k R(d) + \sqrt {k \tilde {\mathcal V}(d)} \Qinv{\epsilon} + \bigo{\log k}, \label{Agamma}
\end{align}
for a properly chosen $\bigo{\log k}$. We will show that $\epsilon^\prime \leq \epsilon$.

 Instead of attempting to compute the infimum in \eqref{Anoisy1} we compute an upper bound to $g(x^k, t)$ by choosing $P_{Z^k}$ individually for each $k$, $t$ and each type $P_{\bar {\mathsf X}}$ of $x^k \in \mathcal T_k$. 
  We let  $P_{Z^k}$ be equiprobable on the conditional type $P_{\bar {\mathsf Z} ^\star | {\mathsf X}}$ that achieves $\mathbb R_{\bar{\mathsf X}; \mathsf Z^\star }(d - \delta )$, formally defined in \eqref{RR(d)g}\footnote{If the probability masses of $P_{\bar {\mathsf Z} ^\star | {\mathsf X}}$ are not divisors of $k$, we take the type closest to $P_{\bar {\mathsf Z} ^\star | {\mathsf X}}$ such that the distortion constraint is not violated. 
 }, where 
\begin{equation}
 \delta = \sqrt{\frac{V_{\bar{\mathsf X} }}{k}} \Qinv{t - \frac{B_{\bar{\mathsf X} }}{\sqrt k}}, 
\end{equation}
and $V_{\bar{\mathsf X} }$, $B_{\bar{\mathsf X}}$ is the normalized variance and the Berry-Esse\'en coefficient of the sum of independent random variables $\sum_{i = 1}^k {\mathsf d} (S_i , z_i)$, where $S_i$ follows the distribution $P_{\mathsf S | \mathsf X = x_i}$. 
If $V_{\bar{\mathsf X}} > 0$, by virtue of the Berry-Esse\'en theorem,
\begin{align}
\pi(x^k; z^k) &= \Prob{ \sum_{i = 1}^k {\mathsf d} (S_i , z_i) > k d ~ | ~ \type{x^k} =  P_{\bar{\mathsf X}} } \\
&\leq t . \label{pibe}
\end{align}
If $V_{\bar{\mathsf X}} = 0$, \eqref{pibe} holds trivially because then ${\mathsf d} (S_i , z_i) = d - \delta$, almost surely.

Denote the following function: 
 \begin{equation}
\mathbb R_{ {\mathsf X}; \mathsf Z^\star}(d) \triangleq \min_{
\substack
{
P_{ \mathsf Z|  {\mathsf X}} \colon \\
\E{\bar {\mathsf d}(  {\mathsf X}, \mathsf Z)} \leq d
}
}
D( P_{\mathsf Z| {\mathsf X}} \| P_{\mathsf Z^\star} | P_{ {\mathsf X}}) \label{RR(d)g},
\end{equation}
where we used the usual notation for conditional relative entropy $D( P_{\mathsf Z| {\mathsf X}} \| P_{\mathsf Z^\star} | P_{ {\mathsf X}}) \triangleq D( P_{\mathsf Z| {\mathsf X}} P_{{\mathsf X}} \| P_{\mathsf Z^\star} P_{{\mathsf X}} )$. 

By the type counting argument and by Taylor's theorem, for all $x^k \in \mathcal T_k$, 
\begin{align}
&~ D(P_{Z^k} \|  P_{\mathsf Z^\star } \times \ldots \times P_{\mathsf Z^\star } ) \notag\\
=&~  k D( P_{\bar {\mathsf Z}^\star | \mathsf X} \|  P_{\mathsf Z^\star} |P_{\bar{\mathsf X}}) + \bigo{\log k} \label{Anoisya} \\
 =&~ k \mathbb R_{\bar {\mathsf X}; \mathsf Z^\star} (d - \delta) + \bigo{\log k}  \\
 =&~ k \mathbb R_{\bar {\mathsf X}; \mathsf Z^\star} (d) + \lambda_{\bar {\mathsf X}}k \delta + \bigo{\log k} \label{Anoisyb}\\
 =&~ \sum_{i = 1}^k  \jmath_{\mathsf X}(x_i, d)  + \lambda_{\mathsf X} k \delta + \bigo{\log k} \label{konto}\\
 =&~ \sum_{i = 1}^k  \jmath_{\mathsf X}(x_i, d)  + \lambda_{\mathsf X} \sqrt{ k V_{\mathsf X} } \Qinv{t}  + \bigo{\log k} \label{Anoisyc},
\end{align}
 where 
\begin{itemize}
\item  \eqref{Anoisya} is by type counting; 
\item \eqref{Anoisyb} where 
\begin{equation}
\lambda_{\bar{\mathsf X}} \triangleq  - \mathbb R_{\bar {\mathsf X}; \mathsf Z^\star}^\prime (d) 
\end{equation}
is by the Taylor theorem, applicable because with finite alphabets and finite distortion measure, $\mathbb R_{\bar {\mathsf X}; \mathsf Z^\star} (d)$ is differentiable with respect to $d$ up to any order \cite{yang1999redundancy}; 
\item \eqref{konto} follows  from the reasoning in \cite[Lemma B.4]{kostina2013thesis}\footnote{Lemma B.4 is proven for $\delta = 0$ but the reasoning still goes through for $ \delta = \bigo{ \sqrt{\frac{\log k}{k}} }$. };
\item \eqref{Anoisyc} is by applying a Taylor expansion to $V_{\bar {\mathsf X}}$  in the vicinity of $P_{\mathsf X}$ and to $\Qinv{t}$ in the vicinity of $t$. Continuous differentiability of  $V_{\bar{\mathsf X}}$
\apxonly{TODO: justify this!} 
as a function of $P_{\bar {\mathsf X}}$ follows from assumption \eqref{item:5} via \eqref{nsc:ingber} and 
\begin{equation}
 \mathbb R_{ \bar {\mathsf X};   {\mathsf Z}^\star}(d) = \mathbb R_{\bar{\mathsf X}}(d) + D(P_{ \bar {\mathsf Z}} \| P_{{\mathsf Z}^\star}),
\end{equation}
where $\bar{\mathsf Z}$ achieves $\mathbb R_{\bar{\mathsf X}}(d) $. 
\end{itemize}

Finally, for scalars $\mu, \gamma$ and for $v > 0$ observe the following.  
\begin{align}
&~ \int_0^1 \1{ \mu + v \Qinv{t} > \gamma} dt \notag\\
=&~ \int_{\infty}^{-\infty} \1{ \mu + v Q(\Qinv{t}) > \gamma} dQ(t) \\
=&~ \frac 1 {\sqrt{2 \pi}} \int_{-\infty}^\infty e^{-\frac{\xi^2}{2}} \1{\mu + v \xi > \gamma }d \xi \\
=&~ \Prob{ \mu + v G > \gamma} \label{Anoisyd},
\end{align}
where $G \sim \mathcal N(0, 1)$. 
Juxtaposing \eqref{Anoisyc} and \eqref{Anoisyd}, we have
\begin{align}
&~ \Prob{  g(X^k, U) > \log \gamma,  U \in I_k, X^k \in \mathcal T_k} 
 \leq
  \int_0^1 dt 
 \\
 \cdot&~\Prob{ \sum_{i = 1}^k \jmath_{\mathsf X}(X_i, d)  + \lambda_{ {\mathsf X}} \sqrt{ k V_{{\mathsf X}} } \Qinv{t} + \bigo{\log k} \geq \log \gamma } \notag\\
=&~ \Prob{ \sum_{i = 1}^k \jmath_{\mathsf X}(X_i, d)  + \lambda_{ {\mathsf X}} \sqrt{ k V_{{\mathsf X}} } G + \bigo{\log k} \geq \log \gamma } \label{Anoisye}.
\end{align}
By the choice of $\gamma$ in \eqref{Agamma} and the Berry-Esse\'en theorem (Theorem \ref{thm:BerryEsseen}), it follows that one can choose $\bigo{\log k} $ so that the right side of \eqref{Anoisyd} is upper bounded by $\epsilon - \frac {3+ 2 |\mathcal A|} {\sqrt k}$. Upon comparison of \eqref{Anoisye} and \eqref{Anoisyweaken}, we conclude that $\epsilon^\prime \leq \epsilon$, as desired.

\bibliographystyle{IEEEtran}
\bibliography{../../ratedistortion}
\end{document}